\documentclass[letterpaper, 10 pt, conference]{ieeeconf}  

\IEEEoverridecommandlockouts                              
\usepackage{amsthm}
\usepackage{balance}
\newtheorem{rem}{\textbf{Remark}}
\newtheorem{lemma}{\textbf{Lemma}}
\newtheorem{prop}{\textbf{Proposition}}

\newtheorem{thm}{\textbf{Theorem}}

\newtheorem{ass}{\textbf{Assumption}}

\newcommand{\proj}{\mathbf{Proj}}

\usepackage{amsmath,amssymb}
\usepackage{graphicx}
\usepackage{amsfonts}
\usepackage{float}
\usepackage{algorithmicx}
\usepackage{algorithm}
\usepackage{algpseudocode}
\usepackage{lipsum}
\usepackage{subfig}
\usepackage{color}
\usepackage{cite}
\usepackage{epstopdf}
\usepackage{amsthm}
\usepackage{calrsfs}
\usepackage{dsfont}
\usepackage[hidelinks]{hyperref}
\newcommand{\gap}{\vspace{0.1cm}}

\DeclareMathAlphabet{\pazocal}{OMS}{zplm}{m}{n}
\newcommand{\newsec}[1]{\gap \noindent {\bf #1}}

\algdef{SE}[DOWHILE]{Do}{doWhile}{\algorithmicdo}[1]{\algorithmicwhile\ #1}%

\title{\LARGE \bf
Duality between density function and value function with applications in constrained optimal control and Markov Decision Process
}

\author{Yuxiao Chen and Aaron D. Ames 
\thanks{Yuxiao Chen Aaron Ames are with the Department of Mechanical and Civil Engineering, California Institute of Technology,
         Pasadena, CA, 91106, USA. Emails:
        {\tt\small \{chenyx, ames\}@caltech.edu}}
}

\begin{document}

\maketitle
\thispagestyle{empty}
\pagestyle{empty}

\begin{abstract}

Density function describes the density of states in the state space of a dynamic system or a Markov Decision Process (MDP). Its evolution follows the Liouville equation. We show that the density function is the dual of the value function in the optimal control problems. By utilizing the duality, constraints that are hard to enforce in the primal value function optimization such as safety constraints in robot navigation, traffic capacity constraints in traffic flow control can be posed on the density function, and the constrained optimal control problem can be solved with a primal-dual algorithm that alternates between the primal and dual optimization. The primal optimization follows the standard optimal control algorithm with a perturbation term generated by the density constraint, and the dual problem solves the Liouville equation to get the density function under a fixed control strategy and updates the perturbation term. Moreover, the proposed method can be extended to the case with exogenous disturbance, and guarantee robust safety under the worst-case disturbance. We apply the proposed method to three examples, a robot navigation problem and a traffic control problem in sim, and a segway control problem with experiment.

\end{abstract}
\section{Introduction}
The problem of optimal control is one of the most well-studied problems in control. Due to Bellman's principle of optimality \cite{bellman2013dynamic}, dynamic programming has been the standard tool for solving optimal control problems, both in continuous state space and discrete state space, such as in the case of a Markov Decision Process (MDP). The key concept in dynamic programming is called a value function, which is the optimal cost-to-go at a given state and time, and the optimal control strategy is derived from the value function. In particular, the value function in a continuous state space follows the Hamilton Jacobi Bellman (HJB) partial Differential Equation (PDE). However, it is difficult to cast some constrained optimal control problems as a constrained value function optimization, such as problems with safety constraint or convergence rate constraint. We show that the density function, first introduced to the control community by Rantzer \cite{rantzer2001dual,prajna2004nonlinear}, is the dual of the value function in many different optimal control formulations, and some constrained optimal control problems can be conveniently written as an optimization over the density function. The idea of duality is not a new one, such as the concept of co-state in classic optimal control \cite{pontryagin2018mathematical}, the occupation measure approach in \cite{lasserre2008nonlinear,majumdar2014convex,zhao2017optimal,zhao2017control}. The costate can be viewed as a special case of the density function in which the density is simply a delta function. The setup for occupation measure does not explicitly use the value function, but rather make the control input as part of the infinite dimensional linear programming, which is approximated by finite dimensional semidefinite programming; while the approach in this paper explicitly uses the fact that the control strategy is determined by the value function, and turn the computation to the HJB PDE of the primal problem.

In \cite{rantzer2001dual}, density function was proposed as the dual to the Lyapunov function to prove stability of nonlinear systems. Since the density function follows the Liouville equation, which is a partial differential equation and is hard to encode, the computation method for the density function has been limited to analytical method (propose one and validate by hand) and Sum of Squares programming \cite{prajna2004nonlinear}. Liouville equation is also directly used to formulate optimal control problem and analytical solution can be found for linear systems, as shown in \cite{brockett2007optimal}. We show in this paper that instead of viewing the density function as a certificate of stability, it actually has physical meaning as the distribution of states, and is the dual of the value function in optimal control. Besides, we propose an Ordinary Differential Equation (ODE) approach to compute the density function, and on top of that a primal-dual algorithm that solves optimal control problems with density constraint. We show the duality relationship and the corresponding primal-dual algorithms for constrained optimal control for both continuous dynamic systems and MDP. Though the formulations are quite different, the math behind them are essentially the same.

The main contributions of this paper are (i) establishes the duality relationship between the value function and the density function in multiple cases, including continuous dynamic systems and Markov Decision Processes (ii) proposes primal-dual algorithms to solve constrained optimal control synthesis for the two cases mentioned previously (iii) discusses extensions of the primal-dual algorithm to allow exogenous disturbance and to include multiple MDPs.

The paper is organized as follows. Section \ref{sec:review} reviews the basic concepts and methods for optimal control and MDP, and the concept of density function in both cases. Section \ref{sec:duality} presents the main result, the duality relationship between the value functions and density functions. Following the duality result, Section \ref{sec:constrained_optimal_control} and Section \ref{sec:MDP_constrained} presents the primal-dual algorithms for constrained optimal control for dynamic systems and MDPs, respectively. Two example cases are presented in Section \ref{sec:result} and finally we conclude in Section \ref{sec:conclusion}.

\textit{Nomenclature} For the remainder of the paper, $\mathbb{N}$ denotes the set of natural numbers, $\mathbb{N}_+$ denotes the positive natural number, $\mathbb{R}^n$ denotes the Euclidean space and $\mathbb{R}^n_{\ge0}$ denotes the positive orthant. Given a dynamic equation $\dot{x}=f(x)$, $\Phi_f(x_0,T)$ denotes the flow map of the dynamics with initial state $x_0$ and horizon $T$. $\left\langle {\cdot,\cdot} \right\rangle$ denotes the inner product, depending on the context. For example, ${\left\langle {a,b} \right\rangle _\pazocal{X}}=\int_\pazocal{X} {a(x) \cdot b(x)dx} $ for two functions $a$ and $b$ defined on $\pazocal{X}$. $\mathbf{0}$ denotes a vector of all zeros or a function that is always zero, depending on the context. $\mathds{1}_S$ denotes the indicator function of a set $S$. A function is of class $C^k$ if it is $k$-th differentiable and the $k$-th derivative (partial derivatives in the case of multivariate functions) is continuous. Given a discrete set $S=\left\{s_1,s_2,...,s_N\right\}$ and a function $h:S\to\mathbb{R}$, $\overline{h}$ denotes the vector form of $h$, and is defined as $[h(s_1),h(s_2),...,h(s_N)]^\intercal$.

\section{Background review and problem setup}\label{sec:review}
In this section, we review the concept of density function and optimal control, and formally define the problem to solve. We will review the optimal control formulation with continuous state and input space as well as the Markov decision process, where the state space and input space are discrete.
\subsection{Optimal control and value function}\label{sec:opt_con_review}
There are numerous results in optimal control, we review the setting with continuous state and input space and continuous time. Given a dynamic system:
\begin{equation}\label{eq:dyn_control}
  \dot{x}=F(x,u),x\in\pazocal{X}=\mathbb{R}^n,u\in\pazocal{U}\subseteq\mathbb{R}^m,
\end{equation}
the standard formulation is the following:
\begin{equation*}\label{eq:opt_con}
  \mathop {\min }\limits_u \int_0^T {C\left( {x\left( t \right),u\left( t \right)} \right)dt}  + D\left( {x\left( T \right)} \right)\;\textrm{s.t.}\; \dot x = F\left( {x,u} \right),
\end{equation*}
where $x\in\pazocal{X}= \mathbb{R}^n$ is the state, $u\in\pazocal{U}\subseteq\mathbb{R}^m$ is the control input, $\dot{x}=F(x,u)$ is the dynamics described as an ODE, $T\in\mathbb{R}_{\ge 0}$ is the horizon of the problem, which can be fixed or dependent on $x$. $C:\pazocal{X}\times\pazocal{U}\to\mathbb{R}$ is the running cost function and $D:\pazocal{X}\to\mathbb{R}$ is the terminal cost function.
\begin{rem}
  For simplicity, we only consider time-invariant dynamics. It should be straightforward to extend the results to the case with time-varying dynamics.
\end{rem}
 The Pontryagin Maximum Principle (PMP) \cite{pontryagin2018mathematical} gives necessary conditions for the optimality of the solution, and the standard tool for solving the optimal control problem is dynamic programming\cite{bellman2013dynamic}, which utilizes the principle of optimality and formulates the problem as a Hamilton Jacobi Bellman PDE:
\begin{equation*}\label{eq:HJB}
  \left\{ {\begin{array}{*{20}{l}}
{\frac{{\partial V}}{{\partial t}} + \mathop {\min }\limits_{u\in\pazocal{U}} \left\{ {\nabla_x V\left( {t,x} \right) \cdot F\left( {x,u} \right) + C\left( {x,u} \right)} \right\} = 0}\\
{V\left( {T,x} \right) = D\left( x \right)}
\end{array}} \right.,
\end{equation*}
where the value function $V:[0,T]\times\pazocal{X}\to\mathbb{R}$ is defined as
\begin{equation*}\label{eq:V_def}
\resizebox{1.0\hsize}{!}{$
\begin{aligned}
  V\left( {t,x_0} \right) = \mathop {\min }\limits_{u(\cdot)\in\pazocal{U}(\cdot)} &\left\{ {\int_t^T {C\left( {x\left( t \right),u\left( t \right)} \right)dt}  + D\left( {x\left( T \right)} \right)} \right\}, \\
  \rm{s.t.}\; &x(t)=x_0,\forall \tau\in[t,T],\dot{x}(\tau)=F(x(\tau),u(\tau))
\end{aligned}
$}
\end{equation*}
which is the optimal cost-to-go of the optimal control problem for an initial condition $x_0$ at time $t$. Once the value function is known, the optimal policy is then
\begin{equation*}\label{eq:opt_policy}
  {\mathbf{u}^ \star }\left( t,x \right) = \mathop {\arg \min }\limits_{u\in\pazocal{U}} \left\{ {\nabla_x V\left( {t,x} \right) \cdot F\left( {x,u} \right) + C\left( {x,u} \right)} \right\}.
\end{equation*}

\subsection{Density function for dynamic systems}\label{sec:density_review}
On the other hand, density function was proposed by Anders Rantzer in \cite{rantzer2001dual} as a dual to Lyapunov function. The density function $\rho:[0,T]\times\pazocal{X}\to\mathbb{R}_{\ge 0}$ can be understood as the measure of state concentration in the state space.
We consider a dynamic system described by
\begin{equation}\label{eq:dyn_eq}
  \dot{x}=f(x),\;x\in\mathbb{R}^n,\;f\in C^1.
\end{equation}
Under \eqref{eq:dyn_eq}, the evolution of a differentiable density function follows the Liouville PDE:
\begin{equation}\label{eq:liouville}
  \begin{aligned}
\frac{{\partial \rho }}{{\partial t}} + \nabla_x  \cdot \left( {\rho\cdot f} \right) &= \phi(t,x,\rho) \\
\rho \left( {0,x} \right) &= {\rho _0}\left( x \right),
\end{aligned}
\end{equation}
where $\nabla_x \cdot (\rho\cdot f)=\nabla_x\rho\cdot f +\nabla_x\cdot f \cdot \rho=\sum\limits_{i=1}^n\frac{\partial f_i}{\partial x_i} \cdot \rho+\sum\limits_{i=1}^{n}\frac{\partial\rho}{\partial x_i}f_i$, $\rho_0:\mathbb{R}^n\to\mathbb{R}_{\ge 0}$ is the initial density function, and $\phi:[0,\infty)\times\pazocal{X}\times\mathbb{R}\to\mathbb{R}$ is the supply function, $\phi(t,x_0,\rho(x_0,t))>0$ means a source, i.e., states with initial condition $x(t)=x_0$ appears at $x_0$ with intensity $\phi(t,x_0,\rho(x_0,t))$, and $\phi(t,x_0,\rho(x_0,t))<0$ denotes a sink, i.e. states exit the system with intensity $\phi(t,x_0,\rho(t,x_0))$ at $x_0$, time $t$. We will later show the condition under which $\rho$ is differentiable. We allow $\phi$ to depend on $\rho$ to allow more flexible characterization of the supply function.
\begin{rem}
  Note the difference between $\phi_+$ and $\rho_0$. $\phi_+$ specifies how new states enter the system over time, while $\rho_0$ specifies the initial distribution of the states at $t=0$.
\end{rem}
\begin{rem}
  The Liouville equation does not require the density function to be nonnegative, but any density function that conveys physical meaning should be nonnegative.
\end{rem}
The Liouville PDE can be transformed and solved as an ODE since
\begin{equation}\label{eq:Liouville_ODE_prep}
  \frac{{\partial \rho }}{{\partial t}} + \nabla_x  \cdot \left( {\rho \cdot f} \right) = {\left. {\frac{{d\rho }}{{dt}}} \right|_{\dot x = f(x)}}+(\nabla_x\cdot f) \rho = \phi.
\end{equation}
This implies that the evolution of the density function alone the trajectory of the dynamic system $\dot{x}=f(x)$ satisfies the following ODE:
\begin{equation}\label{eq:Liouville_ODE}
  \left[ {\begin{array}{*{20}{c}}
{\dot x}\\
{\dot \rho }
\end{array}} \right] = \left[ {\begin{array}{*{20}{c}}
{f\left( {x} \right)}\\
{\phi \left( {t,x,\rho} \right) - \nabla_x  \cdot f\left( {t,x} \right)\rho }
\end{array}} \right].
\end{equation}
Note that here $\rho$ does not have to be differentiable since its evolution along the trajectory of $\dot{x}=f(x)$ is completely determined by $f$ and the initial condition $\rho_0$, and the solutions of an ODE do not intersect.

With \eqref{eq:Liouville_ODE}, we can evaluate the density function at any state $x$, any time $t>0$ with the following two step procedure:
\begin{itemize}
  \item First, solve the reverse ODE of $\dot{x}=-f(x)$ with initial condition $x_t$ to get $\Phi_{-f}(x,t)=\Phi_f(x,-t)$.
  \item Then, solve the extended ODE in \eqref{eq:Liouville_ODE} with initial condition $[\Phi_f(x,-t),\rho_0(\Phi_f(x,-t))]^\intercal$ to time $T$.
\end{itemize}

Define the extended dynamics in \eqref{eq:Liouville_ODE} as $\overline{f}$. Given an initial density function $\rho_0$, from the two-step procedure shown above, we have
     \begin{equation*}\label{eq:rho_rep}
    \rho \left( {t,x} \right) = {\Phi _{\overline f}}{\left( {{{\left[ {{\Phi _f}\left( {x, - t} \right),{\rho _0}\left( {{\Phi _f}\left( {x, - t} \right)} \right)} \right]}^\intercal},t} \right)_{ \downarrow \rho }},
  \end{equation*}
  where $\downarrow \rho$ means the projection of $[x^\intercal,\rho]^\intercal$ to $\rho$.

We give the following theorem for the differentiability of the density function.
\begin{thm}\label{thm:diff_rho}
  Given a dynamic equation in \eqref{eq:dyn_eq} where $f\in C^1$, if $\rho_0\in C^1$, and $\phi\in C^0$, then $\rho\in C^1$.
\end{thm}
To prove Theorem \ref{thm:diff_rho}, we need the following lemma:
\begin{lemma}\label{lemma:diff_initial}
  If $f\in C^k$ in an open neighborhood around $x_0$, then $\forall t>0$, the flow map $\Phi_f\in C^{k+1}$ inside an open neighborhood around $\Phi_f(x_0,t)$.
\end{lemma}
This is adopted from Lemma 2.3 in \cite{teschl2012ordinary}, see the proof therein.
\begin{proof}[Proof of Theorem \ref{thm:diff_rho}]
 Since $f\in C^1$, $\nabla_x\cdot f\in C^0$, and $\phi\in C^0$, therefore $\overline{f}\in C^0$. For any $t>0$, $x_0\in\mathbb{R}^n$, let $x_t = \Phi_f(x_0,t)$, then by Lemma \ref{lemma:diff_initial}, $\frac{d x_t}{d x_0}$ exists and is continuous. Since $f\in C^1$, so is $-f$, therefore $\frac{d x_0}{d x_t}$ exists and is continuous. Since $\overline{f}\in C^0$, with Lemma \ref{lemma:diff_initial}, $\Phi_{\overline{f}}\in C^1$, which implies $\frac{\partial \rho(t,x_t)}{\partial \rho_0(x_0)}$ and $\frac{\partial \rho(t,x_t)}{\partial x_0}$ exist and are continuous. By the chain rule,
  \begin{equation*}
    \frac{{\partial \rho (t,{x_t})}}{{\partial {x_t}}} = \frac{{\partial \rho (t,{x_t})}}{{\partial {\rho _0}({x_0})}}\frac{{d {\rho _0}({x_0})}}{{d {x_0}}}\frac{{d {x_0}}}{{d {x_t}}} + \frac{{\partial \rho (t,{x_t})}}{{\partial {x_0}}}\frac{{d {x_0}}}{{d {x_t}}},
  \end{equation*}
  where all the partial derivatives on the RHS exist and are continuous, therefore $\frac{{\partial \rho (t,{x_t})}}{{\partial {x_t}}}$ exists and is continuous. By \eqref{eq:Liouville_ODE_prep}, $\frac{\partial \rho}{\partial t}=\phi(t,x_t,\rho)-\nabla_x\cdot(\rho\cdot f)$, which exists since $\frac{{\partial \rho (t,{x_t})}}{{\partial {x_t}}}$ exists, and is continuous. Therefore $\rho\in C^1$.
\end{proof}
 For a stationary supply function, i.e. $\phi$ only depending on $x$ and $\rho$, one would hope that there exists a stationary density function that any initial condition converges to, this is not always the case, but we provide sufficient condition for the convergence.

\begin{prop}\label{prop:stationary_rho}
  Given a stationary supply function $\phi$ and an initial density function $\rho_0$, assume that there exists a $\rho_s:\mathbb{R}^n\to\mathbb{R}_{\ge 0}$ such that
   \[\forall x\in\mathbb{R}^n, \phi \left(x, \rho_s  \right) - \nabla  \cdot \left( {{\rho_s } \cdot f} \right) = 0.\]
    For any $x\in\mathbb{R}^n$, if there exists $T\ge0$ such that $\forall t\ge T,\rho_0(\Phi_f(x,-t))=\rho_s(\Phi_f(x,-t))$, then $\forall t\ge T, \rho(t,x)=\rho_s(x)$ .
\end{prop}
\begin{proof}
  The proof follows from the fact that
  \begin{equation*}
    \begin{aligned}
\forall t \ge T,\rho \left( {t,x} \right) &= {\Phi _{\bar f}}{\left( {{{\left[ {{\Phi _f}\left( {x, - t} \right),{\rho _0}\left( {{\Phi _f}\left( {x, - t} \right)} \right)} \right]}^\intercal},t} \right)_{ \downarrow \rho }}\\
 &= {\Phi _{\bar f}}{\left( {{{\left[ {{\Phi _f}\left( {x, - t} \right),{\rho_s }\left( {{\Phi _f}\left( {x, - t} \right)} \right)} \right]}^\intercal},t} \right)_{ \downarrow \rho }}\\
 &= {\rho_s }\left( x \right)
\end{aligned}
  \end{equation*}
\end{proof}
\subsection{Markov decision process}\label{sec:MDP_intro}
The discrete counterpart of optimal control is the Markov decision process (MDP), which is defined as a 4-tuple $(S,A,P_a,R_a)$ where
\begin{itemize}
  \item $S$ is a finite set of states, $N=|S|$ is the number of states,
  \item $A$ is a finite set of actions (sometimes the action at state $s$ is limited to $A_s\subseteq A$), $M=|A|$ is the number of actions,
  \item $P_a(s,s')=\mathbb{P}(s(t+1)=s'\mid s(t)=s,a(t)=a)$ is the transition probability from $s$ to $s'$ under action $a$,
  \item $R_a(s,s')$ is the reward associated with the transition from $s$ to $s'$ under action $a$.
\end{itemize}
A standard MDP optimization solves for the optimal policy that maximizes the discounted cumulative reward
\begin{equation}\label{eq:MDP_cost}
\sum\nolimits_{t = 0}^\infty  {{\gamma ^t}{R_{{a(t)}}}\left( {{s(t)},{s(t + 1)}} \right)},
\end{equation}
where $0\le \gamma\le 1$ is the discount factor.

A policy $\pi$ determines what action to take given the state. A policy is deterministic if it maps a state to a deterministic action, stochastic if it maps a state to a distribution over multiple actions. A policy is stationary if it does not change with time. It can be proved that for a finite MDP with the reward function defined in \eqref{eq:MDP_cost}, there always exists a stationary deterministic policy \cite{puterman2014markov}. In the stochastic policy case, let
\begin{equation*}
\begin{aligned}
  \Delta_s&\doteq\left\{\alpha\in\mathbb{R}_{\ge0}^M\mid\sum\nolimits_{j=1}^M\alpha_j=1,a_j\notin A_{s}\rightarrow \alpha_j=0\right\},\\
  \end{aligned}
\end{equation*}
be the set of all possible stochastic policy at state $s$, which is a probability simplex with entries corresponding to actions not available at $s$ constrained to 0. We let $\pi(s)=[\pi(a_1|s),...,\pi(a_M|s)]\in\Delta_s$ denote the policy at state $s$, where $\pi(a|s)\ge 0$ denotes the probability of taking action $a$ at state $s$, and $\sum\limits_{a\in A_s}\pi(a|s)=1$.

The optimal strategy of MDPs can also be solved by dynamic programming, where it appears as the value iteration method:
\begin{equation*}\label{eq:MDP_value_iteration}
\resizebox{1.0\hsize}{!}{$
\begin{aligned}
\pi(s)& = \mathop {\arg \max }\limits_{\pi(s)\in\Delta_s} \left\{ {\sum\limits_a\pi(a|s)\sum\limits_{s'\in S} {{P_a}\left( s,s' \right)\left( {{R_a}\left( s,s' \right) + \gamma V\left( {s'} \right)} \right)} } \right\}\\
V(s) &= \sum\limits_{a\in A}\sum\limits_{s'\in S} {\pi(a|s) P_a\left( {s,s'} \right)\left( {{R_a}\left( {s,s'} \right) + \gamma V(s')} \right)},
\end{aligned}
$}
\end{equation*}
where $V:S\to\mathbb{R}$ is the value function. For notational convenience, we also use the vector form of $V$, defined as $\overline{V}=[V(s_1),V(s_2),...,V(s_N)]^\intercal\in\mathbb{R}^N$
\subsection{Density function for MDP}
Similarly, one can define the density over states in an MDP as $\rho:S\to\mathbb{R}_{\ge 0}$, and the vector form of the density function is defined as $\overline{\rho}=[\rho(s_1),\rho(s_2),...,\rho(s_N)]^\intercal\in\mathbb{R}^N$. For a given policy $\pi$, let $P^{\pi}$ denote the transition probability matrix, where
\begin{equation}\label{eq:P_pi}
\begin{aligned}
  P_{ij}^\pi  &= \mathbb{P}\left( {{s(t+1)} = {s_j}\mid {s(t)} = {s_i},a\sim\pi \left( {{s_i}} \right)} \right)\\
  &=\sum\limits_{a}\pi(a|s_i)P_a(s_i,s_j).
  \end{aligned}
\end{equation}
Given an initial density $\rho_0$ over states, the evolution of the density under $\pi$ follows
\begin{equation*}\label{eq:density_MDP_evo}
  {\overline{\rho}^{t+1}} = ({P^\pi })^\intercal{\overline{\rho}^t} + \overline{\phi} \left( {{\overline{\rho}^t}} \right),
\end{equation*}
where $\phi:S\times\mathbb{R} \to \mathbb{R}$ is the supply function, and $\overline{\phi}(\overline{\rho})=[\phi(s_1,\rho(s_1)),\phi(s_2,\rho(s_2)),...,\phi(s_N,\rho(s_N))]^\intercal$ is its vector form. Here we only consider stationary supply functions, i.e., $\phi$ that do not depend on $t$.
\begin{rem}
   The constraint $\mathbf{1}^\intercal\overline{\rho}^t = 1$ is not enforced as in the case of probability distribution. The probability distribution can be viewed as a special case of density with $\mathbf{1}^\intercal\overline{\rho}^t = 1$ and $\phi=\mathbf{0}$.
\end{rem}
\begin{rem}
  The idea of the dual variable in MDP has been studied, for example, in \cite{altman1999constrained}, and recently in \cite{dai2017boosting}, but they differ from the density function discussed in this paper in that the density function has a physical meaning rather than simply being the dual variable. The evolution of density function is governed by not only the Liouville equation, but also the supply function and the initial condition. Therefore, we can pose constraints on the density function with physical meaning.
\end{rem}
\section{Duality in optimal control of continuous dynamic systems}\label{sec:duality}
In this section, we show the main results of this paper, the duality relationship between the value function and the density function in the continuous optimal control problems and optimal policy synthesis for MDPs.

Depending on the setup of the optimal control problem, the supply function $\phi$ will take different forms. We present several setups, including optimal control problems in the Bolza form with terminal constraint, the infinite horizon discounted case, and the fixed finite horizon case. However, since the derivations are similar, we only present the detail of the first setup, that is, the optimal control problem with terminal constraint. The derivation of the duality result for this setup is the most complicated among the three setups.

For simplicity, we only consider the time invariant case, i.e., the dynamics is time invariant and the supply function $\phi$ is also time invariant, only depending on the state $x$ and the density $\rho$.

\subsection{Duality in optimal control with terminal condition}\label{sec:opt_control_destination}

Given the dynamic system in \eqref{eq:dyn_control}, consider the optimal control problem in the Bolza form with terminal condition $x\in \pazocal{X}_g$, where $\pazocal{X}_g$ is the destination:
\begin{equation*}\label{eq:opt_con_destination}
  \mathop {\min }\limits_u \int_0^T {C\left( {x\left( t \right),u\left( t \right)} \right)dt}  + D\left( {x\left( T \right)} \right)\;\textrm{s.t.}\; \dot x = F\left( {x,u} \right),
\end{equation*}
where $T$ is the first time that the state reaches $\pazocal{X}_g$. We further assume that
\begin{equation*}\label{eq:U_def}
  \pazocal{U}\doteq\left\{u\mid g(u)\le 0 \right\}
\end{equation*}
for some function $g$.

In this case, it can be shown that the value function $V$ only depends on the state $x$, so is the optimal control strategy $\mathbf{u}^\star$.

We consider a supply function $\phi=\phi_+ + \phi_-$, where $\phi_+$ is a stationary nonnegative supply function that only depends on $x$, encoding the information of the distribution of new states entering the state space, and $\phi_-$ takes the following form:
\begin{equation}\label{eq:phi_-}
  {\phi _ - }\left( {t,x,\rho } \right) = \left\{ {\begin{aligned}
&0,&x &\notin {\pazocal{X}_g}\\
 &- \delta \left( t \right)\rho ,&x &\in {\partial\pazocal{X}_g}\\
 &-\phi_+(x), & x &\in {\mathrm{int}(\pazocal{X}_g)},
\end{aligned}} \right.
\end{equation}
where $\partial\pazocal{X}_g$ and $\mathrm{int}(\pazocal{X}_g)$ are the boundary and interior of $\pazocal{X}_g$, $\delta(t)$ is a Dirac delta function at $t$.
This means that the density function immediately becomes zero once the state enters $\pazocal{X}_g$.

An example of this setup is the robot navigation problem where the initial position of the robot follows the distribution $\phi_+$ and the goal is to reach the destination. Another example is the mail collection problem where the destination is the post office and new mails to be collected pop up following the distribution of $\phi_+$.

\begin{ass}\label{ass:inv}
  We assume that there exists a compact set $\pazocal{S}\subseteq\pazocal{X}$ such that for the system in \eqref{eq:dyn_control} under all control strategies, supply function $\phi$, and $\rho_0$ considered in this paper, $\forall t>0, x(t)\in\pazocal{S}$. Clearly, $\forall x\notin\pazocal{S},\phi_+(x)=\rho_0(x)=0$. Furthermore, we assume that $\pazocal{X}_g$ is a compact set and all state in $\pazocal{S}$ can reach $\pazocal{X}_g$ in finite time.

\end{ass}
This is a reasonable assumption since in a typical optimal control problem, the optimal state evolution is bounded. Together with Proposition \ref{prop:stationary_rho}, if $\rho_s$ exists and $\Phi_f(x,-t)\notin\pazocal{S}$, $\rho(t,x)=\rho_s(x)$ since $\rho_s(\Phi_f(x,-t))=\rho_0(\Phi_f(x,-t))=0$.

Instead of the optimal control problem for an individual initial condition, we look at the overall cost of the whole system, similar to \cite{rantzer2001dual}. Given a control strategy $u=\mathbf{u}(x)$, let $V^\mathbf{u}$ be the cost-to-go associated to $\mathbf{u}$ for a given state, then the overall cost rate over time is
\begin{equation*}\label{eq:overall_cost_V}
  J_p = {\left\langle {{\phi _ + },{V^{\bf{u}}}} \right\rangle _\pazocal{X}},
\end{equation*}
This is interpreted as the inner product between the value function and the positive supply. Indeed, for every new state entering the state space at $x$, it incurs a total cost of $V^\mathbf{u}(x)$, and the new states emerge following rate $\phi_+$ across the state space, thus ${\left\langle {{\phi _ + },{V^{\bf{u}}}} \right\rangle _\pazocal{X}}$ is the total cost rate.

By Bellman's principle of optimality, we know that the optimal value function of each state is independent of the state trajectory before the it reaches that state, which implies that there exists a pure state feedback law $\mathbf{u}^\star$ that minimizes the overall cost $J_p$, and is determined by the following stationary HJB equation:
\begin{equation}\label{eq:primal_opt}
  \begin{aligned}
{J_p^ \star } = &{\left\langle {{\phi _ + },{V^{{{\mathbf{u}}^ \star }}}} \right\rangle _\pazocal{X}}\\
\rm{s.t.}&{\mathbf{u}^ \star }(x) = \mathop {\arg \min }\limits_{u \in \pazocal{U}} \;\nabla V \cdot F(x,u)+C(x,u)\\
&C_{\mathbf{u}^\star} + \nabla V \cdot F_{\mathbf{u}^\star} = 0\\
&V{\mid _{{\pazocal{X}_g}}} = D,
\end{aligned}
\end{equation}
where $C_{\mathbf{u}^\star}(x)\doteq C(x,\mathbf{u}^\star(x))$ and $F_{\mathbf{u}^\star}(x)\doteq F(x,\mathbf{u}^\star(x))$.
Since $V^{\mathbf{u}^\star}$ is completely determined by the equality constraint, we leave out the optimization sign. This is referred to as the primal problem.

Alternatively, if the density function reaches a stationary density $\rho_s$, the overall cost rate can also be represented as
\begin{equation}\label{eq:overall_cost_rho}
  J_d ={\left\langle {{C_\mathbf{u}},{\rho_s}} \right\rangle _\pazocal{X}}+{\left\langle {{D},{-\phi_-}} \right\rangle _\pazocal{X}},
\end{equation}
where the first part represents the overall running cost and the second part represents the terminal cost. Indeed, the density function represents the presence of the states in the state space, and \eqref{eq:overall_cost_rho} is simply the inner product of the density with the cost function, including the running cost and terminal cost. With this formulation, the following optimization solves for the optimal overall cost rate:
\begin{equation}\label{eq:dual_opt}
    \begin{aligned}
\mathop {\min }\limits_{\rho_s ,\mathbf{u}} &{\left\langle {\rho_s ,C_{\mathbf{u}}} \right\rangle}_{\pazocal{X}}-{\left\langle {\phi_- ,D} \right\rangle}_\pazocal{X}&& \\
\rm{s.t.}&\nabla  \cdot \left( {\rho_s \cdot F_{\mathbf{u}}) } \right) = \phi ,&(a)&\\
   &\rho_s(x)  \ge 0 ,&(b)&\\
   &\forall x \in\pazocal{X},g(\mathbf{u}(x))\le0 ,&(c)&
\end{aligned}
\end{equation}
where $(a)$ is the Liouville equation, $(b)$ requires that the density stays nonnegative, $(c)$ requires that the input is within the bound $\pazocal{U}$. This is referred to as the dual problem.

The main theorem of this paper states that the value function and the stationary density function are dual to each other.
\begin{thm}\label{thm:opt_con_duality}
  For the dynamic system in \eqref{eq:dyn_control}, the optimization in \eqref{eq:dual_opt} and \eqref{eq:primal_opt} are dual to each other. If there exist unique and differentiable $V$ and $\rho_s$ that are optimal solutions to the two problems, there is no duality gap.
\end{thm}
Before presenting the proof, we need the following lemma:
\begin{lemma}\label{lemma:phi_balance}
  Under Assumption \ref{ass:inv}, if the system reaches a stationary density distribution,
  \[\int_\pazocal{X} {\phi dx} =0.\]
\end{lemma}
\begin{proof}
  Under Assumption \ref{ass:inv}, By the adjoint condition:
  \begin{equation*}\label{eq:inv_density}
  \resizebox{1.0\hsize}{!}{$
    \int_\pazocal{X} {\nabla  \cdot \left( {\rho_s  \cdot F_\mathbf{u}} \right)dx}  = \int_\pazocal{S} {\nabla  \cdot \left( {\rho_s  \cdot F_\mathbf{u}} \right)dx} \int_{\partial \pazocal{S}} {\rho_s F_\mathbf{u} \cdot \overrightarrow n \;ds}  = 0.
    $}
  \end{equation*}
  Then \[\int_\pazocal{X} {\phi dx}=\int_\pazocal{X} {  - \nabla  \cdot \left( {\rho_s  \cdot F_\mathbf{u}} \right)dx}    = 0.\]
\end{proof}
\begin{proof}[Proof of Theorem \ref{thm:opt_con_duality}]
We show one direction, from \eqref{eq:dual_opt} to \eqref{eq:primal_opt}, and the other direction is similar. The Lagrangian is formulated as
\begin{equation}\label{eq:dual_Lag}
\begin{array}{c}
  \pazocal{L} = {\left\langle {{\rho _s},C_{\mathbf{u}}} \right\rangle}_\pazocal{X} - {\left\langle {{\phi _ - },D} \right\rangle _\pazocal{X}} + {\left\langle {\mu ,\phi  - \nabla  \cdot \left( {{\rho _s} \cdot F_{\mathbf{u}}} \right)} \right\rangle _\pazocal{X}} \\
  - {\left\langle {{\lambda_0}, - {\rho _s}} \right\rangle _\pazocal{X}}-{\left\langle {{\lambda_1}, g\circ\mathbf{u}} \right\rangle _\pazocal{X}},
\end{array}
\end{equation}
where $\mu:\pazocal{X}\to\mathbb{R}$, $\lambda_0:\pazocal{X}\to\mathbb{R}_+$ and $\lambda_1:\pazocal{U}\to\mathbb{R}_+$ are the Lagrange multipliers for constraint (a), (b), and (c) in \eqref{eq:dual_opt}, respectively.
First notice that 
\[ - {\left\langle {{\phi _ - },D} \right\rangle _\pazocal{X}} = {\left\langle {  \delta {\rho _s},D} \right\rangle} _{\partial\pazocal{X}_g}+{\left\langle {  \phi_+ ,D} \right\rangle} _{\mathrm{int}(\pazocal{X}_g)}\]
 Then by Assumption \ref{ass:inv}, we use the adjoint condition:
\begin{equation*}\label{eq:adjoint}
  {\left\langle {\mu ,\nabla  \cdot \left( {\rho_s  \cdot F_{\mathbf{u}}} \right)} \right\rangle}_\pazocal{X}  =  - {\left\langle {\nabla \mu ,\rho_s  \cdot F_{\mathbf{u}}} \right\rangle}_\pazocal{X}  =  - {\left\langle {\rho_s ,\nabla \mu  \cdot F_{\mathbf{u}}} \right\rangle}_\pazocal{X}.
\end{equation*}
The Lagrangian then becomes
\begin{equation}\label{eq:dual_lag2}
\begin{aligned}
  \pazocal{L} &= {\left\langle {{\rho _s}, C_{\mathbf{u}} +\mathds{1}_{\partial\pazocal{X}_g}\delta D + \nabla \mu F_{\mathbf{u}} + \lambda_0} \right\rangle}_\pazocal{X}\\
  &+{\left\langle {  \phi_+ ,D} \right\rangle} _{\mathrm{int}(\pazocal{X}_g)}
  + {\left\langle {\mu ,\phi } \right\rangle}_\pazocal{X}  - {\left\langle {{\lambda _1},g\circ \mathbf{u}} \right\rangle}_\pazocal{X}
\end{aligned}
\end{equation}
The Kuhn-Karush-Tucker (KKT) condition reads

Stationarity condition:
\begin{equation*}\label{eq:stationarity}
\begin{aligned}
  \frac{{\partial \pazocal{L}}}{{\partial \rho_s }} &=& C_\mathbf{u} +\mathds{1}_{\partial\pazocal{X}_g}\delta D + \nabla \mu F_\mathbf{u} + \lambda_0 &=& 0 \\
  \frac{{\partial \pazocal{L}}}{{\partial u}} &=& \rho_s \left( \frac{{\partial C}}{{\partial u}} + \nabla \mu {\frac{{\partial F}}{{\partial u}} }+\lambda_1{\frac{{\partial g}}{{\partial u}} } \right) &=& 0
  \end{aligned}
\end{equation*}

Complementary slackness:
\begin{equation*}\label{eq:complementary_slackness}
\mu \cdot(\phi-\nabla  \cdot \left( {\rho_s  \cdot F} \right)) =  {\lambda _0}\cdot\rho_s   = \lambda_1\cdot g(u)  =  0
\end{equation*}
This implies that when $\rho_s>0$, i.e. for area in $\pazocal{X}$ with nonzero density,
\begin{equation*}\label{eq:dual_optimality_condition}
  \begin{array}{c}
{\mathbf{u}^ \star(x) } = \mathop {\arg \min }\limits_{g(u)\le 0} \;C(x,u) + \nabla \mu\cdot F(x,u),\\
C_\mathbf{u^\star} + \delta(t) D\mathds{1}_{\partial\pazocal{X}_g} +\nabla \mu F_\mathbf{u^\star} = 0,
\end{array}
\end{equation*}
 which directly come from the stationarity condition and utilized the fact that $\rho_s>0\to \lambda_0=0$.
  Furthermore, since $\rho_s=0$ inside $\pazocal{X}_g$, the optimal input can be picked arbitrarily from $\pazocal{U}$, therefore $\mu$ has to be constant within $\mathrm{int}(\pazocal{X}_g)$. let $\mu_0=\mu(x)\mid_{\mathrm{int}(\pazocal{X}_g)}$. Note that by Lemma \ref{lemma:phi_balance}, at stationary density,
  $\int_\pazocal{X} {\phi dx} =0$, which implies
  \[{\left\langle {\mu  - {\mu _0},\phi } \right\rangle _\pazocal{X}}= {\left\langle {\mu  ,\phi } \right\rangle _\pazocal{X}}-{\left\langle { {\mu _0},\phi } \right\rangle _\pazocal{X}}={\left\langle {\mu ,\phi } \right\rangle _\pazocal{X}} .\]
  Therefore, we can replace $\mu$ with $\mu-\mu_0$ and both the KKT condition and the value of the Lagrangian remain unchanged. Without loss of generality, we can assign $\mu_0=0$. Then $\mu$ satisfies
\begin{equation*}\label{eq:mu_property}
  \begin{array}{l}
\mu  = \left\{ \begin{aligned}
&0,   &x &\in \mathrm{int}(\pazocal{X}_g)\\
&D(x),&x &\in \partial {\pazocal{X}_g}
\end{aligned} \right.,\\
\forall x\notin \pazocal{X}_g,\nabla \mu  \cdot F_{\mathbf{u}^\star} =  - C_{\mathbf{u}^\star}
\end{array}
\end{equation*}
Define
\begin{equation*}\label{eq:dual_V_def}
  V=\mu+\mathds{1}_{\mathrm{int}(\pazocal{X}_g)}D,
\end{equation*}
then we have
\begin{equation*}\label{eq:dual_V_res}
  \begin{array}{l}
V{\mid _{{\pazocal{X}_g}}} = D,\\
\forall x \notin {\pazocal{X}_g},\nabla \mu  \cdot F_{\mathbf{u}^\star} =  - C_{\mathbf{u}^\star},\\
\mathbf{u}^ \star(x)  = \mathop {\arg \min }\limits_{u\in\pazocal{U}} \;C(x,u) + \nabla V \cdot F(x,u),
\end{array}
\end{equation*}
which is exactly the solution of the primal problem in \eqref{eq:primal_opt}.

Besides, from \eqref{eq:dual_lag2}, if such an solution to the optimal problem exists, the dual objective becomes
\begin{equation*}
  {J_d^ \star } = \mathop {\max }\limits_{{\lambda _0},{\lambda _1},\mu }\mathop {\min }\limits_{\rho_s,\mathbf{u} } \pazocal{L} = {\left\langle {{\phi _ + },D} \right\rangle _{{\mathop{\rm int}} ({\pazocal{X}_g})}} + {\left\langle {\phi ,\mu } \right\rangle _\pazocal{X}}.
\end{equation*}
Since $\phi_-\mid_{\pazocal{X}\backslash\pazocal{X}_g}=0$,
\begin{equation*}\label{eq:duality_gap}
  {J_d^ \star }= {\left\langle {\phi_+ ,V } \right\rangle _\pazocal{X}}=J_p^\star,
\end{equation*}
 which shows that there is no duality gap.
\end{proof}

\subsection{Duality in several other forms of optimal control}
Consider an infinite horizon optimal control problem with the following cost function:
\begin{equation*}\label{eq:discount_cost}
  V(x) = \int_0^\infty  {{e^{ - \kappa \tau }}C\left( {x\left( \tau  \right),u\left( \tau  \right)} \right)d\tau },
\end{equation*}
where $\kappa$ is the discount factor.
In this case, the negative supply function takes the following form:
\begin{equation*}\label{eq:negative_supply_discount}
  \forall x\in\pazocal{X}, \phi_-(x,\rho)=-\kappa \rho.
\end{equation*}
The primal optimal control problem is the following:
\begin{equation}\label{eq:primal_discount}
\begin{aligned}
J_p^\star=&\left\langle {V,{\phi _ + }} \right\rangle \\
\rm{s.t.}&C_{\mathbf{u}^\star} + \nabla V \cdot F_{\mathbf{u}^\star} -\kappa V = 0\\
&{\mathbf{u}^ \star }(x) = \mathop {\arg \min }\limits_{ u\in\pazocal{U}} \;C(x,u) + \nabla V \cdot F(x,u).
\end{aligned}
\end{equation}
The corresponding density optimization takes the form
\begin{equation}\label{eq:dual_opt_discount_factor}
    \begin{aligned}
J_d^\star=\mathop {\min }\limits_{\rho_s ,\mathbf{u}} &{\left\langle {\rho_s ,C_\mathbf{u}} \right\rangle}_{\pazocal{X}}\\
\rm{s.t.}&\nabla  \cdot \left( {\rho_s \cdot F(x,\mathbf{u}(x)) } \right) = \phi_+-\kappa\rho_s ,\\
   &\forall x\in \pazocal{X},g(\mathbf{u}(x))\le0 ,\;\rho_s(x)  \ge 0.
\end{aligned}
\end{equation}
The duality between \eqref{eq:primal_discount} and \eqref{eq:dual_opt_discount_factor} is proved in Theorem 1 in \cite{chen2019optimal}, see the proof therein.

Another setup is a fixed horizon optimal control problem. In this case, there is no supply function or stationary density, but instead an initial density of the states $\rho_0$. The following cost function is considered:
\begin{equation*}\label{eq:V_fixed_horizon}
\resizebox{1.0\hsize}{!}{$
  V(t,x_t) = \int_0^T  {C\left( {x\left( \tau  \right),u\left( \tau  \right)} \right)d\tau } +D(x(T)),x(t)=x_t,\dot{x}=F(x,u),
$}
\end{equation*}
where $T$ is fixed.

The primal optimal control problem is the following:
\begin{equation*}\label{eq:primal_opt_fixed_horizon}
    \begin{aligned}
J_p^\star =&{\left\langle {V\left( {0, \cdot } \right),\rho_0} \right\rangle _\pazocal{X}}\\
\rm{s.t.}&\frac{{\partial V}}{{\partial t}} +C_{\mathbf{u}^\star} + \nabla_x V\cdot F_{\mathbf{u}^\star} = 0\\
&{\mathbf{u}^ \star }(t,x) = \mathop {\arg \min }\limits_{u\in\pazocal{U}} C(x,u) + \nabla_x V(t,x) \cdot F(x,u)\\
&V(T,\cdot)=D.
\end{aligned}
\end{equation*}
In this case $\rho:[0,T]\times\pazocal{X}\to\mathbb{R}$ and $\mathbf{u}^\star:[0,T]\times\pazocal{X}\to\mathbb{R}$ depend on both $x$ and $t$, and we can show that the dual problem to this is
\begin{equation*}\label{eq:dual_opt_fixed_horizon}
  \begin{aligned}
\mathop {\min }\limits_{\rho ,\mathbf{u}} \;&{\left\langle {\rho ,C_{\mathbf{u}}} \right\rangle _{[0,T]\times\pazocal{X}}}+{\left\langle {\rho(T,\cdot) ,D} \right\rangle _{\pazocal{X}}}\\
\rm{s.t.}&\frac{{\partial \rho }}{{\partial t}} + \nabla_x  \cdot \left( {\rho \cdot F_\mathbf{u}} \right) = 0,\\
&\rho \left( {0, \cdot } \right) = {\rho _0},\;\rho\ge0\\
\end{aligned}
\end{equation*}
Similar duality result can be proved, and is omitted here.
\subsection{Duality in MDP optimization}\label{sec:duality_MDP}
Similar to the optimal control problem in continuous state and input space, there is duality relationship between the density function and the value function in MDP. We consider two setups of MDP, one with infinite horizon and discounted reward, and one with a sink where states vanish after reaching the sink.

First, we discuss the discounted reward case. Again, we only consider the time invariant case. For the MDP introduced in Section \ref{sec:MDP_intro}, the value function for an individual state $s_0$ is defined as
\[ V(s_0)=\sum\nolimits_{t = 0}^\infty  {{\gamma ^t}{R_{{a(t)}}}\left( {{s(t)},{s(t+1)}} \right)}\;\mathrm{s.t.}\; s(0)=s_0.\]

In the time invariant case, it is assumed that initial states appear at each state following the positive supply $\phi_+:S\to\mathbb{R}$. Recall the solution of the optimal policy for the MDP reviewed in Section \ref{sec:MDP_intro}, the overall reward for an MDP given the positive supply $\phi_+$ is then $\left\langle {{\phi_+ },V} \right\rangle _S$, and the primal optimization is:
\begin{equation}\label{eq:overall_reward_MDP}
\resizebox{1\hsize}{!}{$
\begin{aligned}
  J_p^\star &= {\left\langle {{\phi_+ },V} \right\rangle _S},\\
  \mathrm{s.t.}\;V\left( s \right) &= \sum\limits_{a\in A}\sum\limits_{s'\in S} {\pi(a|s) P_a\left( {s,s'} \right)\left( {{R_a}\left( {s,s'} \right) + \gamma V\left( {s'} \right)} \right)},\\
   \pi(s)& = \mathop {\arg \max }\limits_{\pi(s)\in\Delta_s} \left\{ {\sum\limits_a\pi(a|s)\sum\limits_{s'\in S} {{P_a}\left( s,s' \right)\left( {{R_a}\left( s,s' \right) + \gamma V\left( {s'} \right)} \right)} } \right\}
  \end{aligned}
  $}
\end{equation}

On the density side, the corresponding supply function for this setup is
\[\overline{\phi}(\overline{\rho})=\overline{\phi}_+-(1-\gamma)(P^\pi)^\intercal\overline{\rho}\]
Then the corresponding Liouville equation for density vector is the following:
\begin{equation*}\label{eq:MDP_liouville_discount}
  \overline{\rho}^{t+1} = (P^\pi)^\intercal \overline{\rho}^t + \overline{\phi}_+-(1-\gamma)(P^\pi)^\intercal \overline{\rho}^t=\gamma (P^\pi)^\intercal \overline{\rho}^t+\overline{\phi}_+,
\end{equation*}
The stationary density ${\rho}_s:S\to\mathbb{R}_{\ge 0}$ then satisfies
\[\overline{\rho}_s=\gamma (P^\pi)^\intercal\overline{\rho}_s+\overline{\phi}_+,\]
where $\overline{\rho}_s$ is the vector form of $\rho_s$.
For a $0<\gamma<1$, since $P^\pi$ defined in \eqref{eq:P_pi} is a probability transition matrix, $(I-\gamma (P^\pi)^\intercal)$ is nonsingular, thus the stationary $\overline{\rho}_s=(I-\gamma (P^\pi)^\intercal)^{-1}\overline{\phi}_+$ always exists and is unique.

The dual optimization is then
\begin{equation}\label{eq:MDP_dual}
  \begin{aligned}
J_d^\star=\max\limits_{\rho_s,\pi} &\sum\limits_s {{\rho(s)}\sum\limits_{a}\pi(a|s)\sum\limits_{s'}P_a(s,s') R_a(s,s')} \\
\rm{s.t.}\;&\overline{\rho}_s=\gamma (P^\pi)^\intercal \overline{\rho}_s+\overline{\phi}_+,\forall s,\pi(s)\in\Delta_s.
\end{aligned}
\end{equation}
\begin{thm}\label{thm:MDP_duality}
  The primal problem in \eqref{eq:overall_reward_MDP} and dual problem in \eqref{eq:MDP_dual} are dual to each other with no duality gap.
\end{thm}
\begin{proof}
  Starting with the dual problem in \eqref{eq:MDP_dual}. The Lagrangian is then formulated as
\begin{equation*}
\resizebox{.95\hsize}{!}{$
  \begin{aligned}
\pazocal{L} &= \sum\limits_s \rho_s(s)\sum\limits_a \pi(a|s)\sum\limits_{s'}P_a(s,s')R_a(s,s')\\
&+\sum\limits_s\mu(s)\left(\rho_s(s)-\gamma\sum\limits_{s'}\sum\limits_a\pi(a|s')P_a(s',s)\rho_s(s')+\phi_+(s)\right)
\end{aligned}
$}
\end{equation*}
Rearranging the terms, we get
\begin{equation*}
\begin{aligned}
  \pazocal{L}&=\sum\limits_s\rho_s(s)\left(
    \begin{array}{c}
    \sum\limits_a\pi(a|s)\sum\limits_{s'}P_a(s,s')R_a(s,s') \\
    +\mu(s)-\gamma\sum\limits_a\sum\limits_{s'}\pi(a|s)P_a(s,s')\mu(s')
  \end{array}
  \right)\\
  &+\sum\limits_s \mu(s) \phi_+(s)
  \end{aligned}
\end{equation*}
Replacing $\mu$ with $V$, the KKT condition implies
\begin{equation*}
\resizebox{1\hsize}{!}{$
\begin{aligned}
{V(s)}& = \sum\limits_a\pi(a|s)\sum\limits_{s'}(P_a(s,s')R_a(s,s')+\gamma V(s')) \\
\pi(s)& = \mathop {\arg \max }\limits_{\pi(s)\in\Delta_s} \left\{ {\sum\limits_{a}\pi(a|s)\sum\limits_{s'\in S} {{P_a}\left( s,s' \right)\left( {{R_a}\left( s,s' \right) + \gamma V(s')} \right)} } \right\},
\end{aligned}
$}
\end{equation*}
which is the optimality condition for the value function, and it's easy to check that there is no duality gap, i.e.,
\begin{equation*}
  J_d^\star=\mathop {\min }\limits_{\rho,\pi}  \pazocal{L} = \sum\limits_s {\phi_+(s) V(s)}=J_p^\star.
\end{equation*}
\end{proof}
On top of the first setup, the second setup adds a sink set $S_-$, where the state vanishes once inside $S_-$. This corresponds to the situation where the goal is to travel to the goal state (the sink $S_-$), and the state vanishes once it arrives at the goal state.
The primal optimization is
\begin{equation}\label{eq:primal_MDP_sink}
  \begin{aligned}
  J_p^\star &= {\left\langle {{\phi_+ },V} \right\rangle _S},\\
  \mathrm{s.t.}\;&\forall s\in S_-,V(s)=0,\forall s\notin S_-,\\
   V\left( s \right) &= \sum\limits_{a\in A}\sum\limits_{s'\in S} {\pi(a|s) P_a\left( {s,s'} \right)\left( {{R_a}\left( {s,s'} \right) +  \gamma V\left( {s'} \right)} \right)},\\
   \pi(s)& = \mathop {\arg \max }\limits_{\pi(s)\in\Delta_s} \left\{ {\sum\limits_{s'\in S} {{P_a}\left( s,s' \right)\left( {{R_a}\left( s,s' \right) + V\left( {s'} \right)} \right)} } \right\}
  \end{aligned}
\end{equation}
In this case, the supply function is defined as
\begin{equation}\label{eq:supply_sink}
  \overline{\phi}(\overline{\rho})=\overline{\phi}_+ -((1-\gamma)I+\gamma\mathbf{diag}(\mathds{1}_{S_-}))(P^\pi)^\intercal\overline{\rho},
\end{equation}
where $\mathbf{diag}(\mathds{1}_{S_-})\in \mathbb{R}^{N\times N}$ is a diagonal matrix whose diagonal vector is $\mathds{1}_{S_-}\in \mathbb{R}^N$,
\[\mathds{1}_{S_-}(i)=\left\{\begin{aligned}&1,&s_i\in S_-\\
&0,&s_i\notin S_-
\end{aligned}\right..
\]
Given $P^\pi$, define the cropped transition probability matrix $\tilde{P}^\pi$ by modifying the columns of $P^\pi$ corresponding to the states in $S_-$ to all zeros. Then it is easy to check that
\[P^\pi-((1-\gamma)I+\gamma\mathbf{diag}(\mathds{1}_{S_-}))P^\pi=\tilde{P}^\pi.\]
Then the Liouville equation in this case is simply
\begin{equation*}
  \overline{\rho}^{t+1}=\gamma(\tilde{P}^\pi)^\intercal \overline{\rho}^t+\overline{\phi}_+.
\end{equation*}
The dual optimization is then
\begin{equation}\label{eq:MDP_dual_sink}
  \begin{aligned}
J_d^\star=\max\limits_{\rho_s,\pi} &\sum\limits_{s\notin S_-} {{\rho_s(s)}\sum\limits_{a}\pi(a|s)\sum\limits_{s'}P_a(s,s') R_a(s,s')} \\
\rm{s.t.}\;&\overline{\rho}_s= \gamma(\tilde{P}^\pi)^\intercal \overline{\rho}_s+\overline{\phi}_+,\forall s\notin S_-,\pi(s)\in\Delta_s.
\end{aligned}
\end{equation}
\begin{thm}
  When the optimal solution to the primal optimization in \eqref{eq:primal_MDP_sink} and the dual optimization in \eqref{eq:MDP_dual_sink} exist and are finite, the two problems are dual to each other, and there is no duality gap.
\end{thm}
The proof is similar to the proof of Theorem \ref{thm:MDP_duality}, and is omitted here.
\section{Constrained optimal control for dynamic systems}\label{sec:constrained_optimal_control}
With density functions, it is convenient to pose some constrained optimal control problems that are hard to pose with value function. Here we list a few.
\begin{itemize}
  \item In an optimal control problem, preventing the state to enter the danger area $\pazocal{X}_d$,
  \item In a robot navigation problem solved as a finite-horizon optimal control problem, enforcing a lower bound on the portion of states that reach the destination at the end of the horizon,
  \item In the macroscopic traffic assignment problem, enforcing upper bounds on road sections to prevent congestion.
\end{itemize}
All of the above mentioned problems can be posed as constrained optimizations on the density function, and we present a primal-dual algorithm to solve them.

We will show two examples, the first one is the optimal control problem with a terminal condition, and the constraint is that the states never enter a danger area $\pazocal{X}_d$. The second example is an MDP with upper bounds on the density of some states.
\subsection{Optimal control with safety constraint}
For the optimal control problem with a terminal condition, the unconstrained version is reviewed in Section \ref{sec:opt_con_review} and we present the duality result in Section \ref{sec:opt_control_destination}. Although the safety constraint is hard to impose on the value function, it is very convenient to impose it on the density formulation. The constrained optimization of density function is
\begin{equation}\label{eq:constrained_opt_con_dual}
  \begin{aligned}
\mathop {\min }\limits_{\rho_s ,\mathbf{u}} &{\left\langle {\rho_s ,C_\mathbf{u}} \right\rangle}_{\pazocal{X}}-{\left\langle {\phi_- ,D} \right\rangle}_\pazocal{X}\\
\rm{s.t.}&\nabla  \cdot \left( {\rho_s \cdot F_\mathbf{u}) } \right) = \phi ,\\
   &\forall x\in\pazocal{X},g(\mathbf{u}(x))\le0 ,\;\rho_s(x)  \ge 0\\
   &\rho _s\mid_{\pazocal{X}_d} \le \rho^{\max},
\end{aligned}
\end{equation}
where $\rho^{\max}$ is the tolerance, and it takes the value 0 if the constraint is absolute.

This optimization on density function is an infinite dimensional linear programming with equality constraint, which is hard to solve. However, one can use a primal-dual algorithm and solve the primal value function problem instead. With this extra safety constraint, the Lagrangian becomes
\begin{equation}\label{eq:constrained_lag}
  \begin{array}{c}
  \pazocal{L} = {\left\langle {{\rho _s},C_{\mathbf{u}}} \right\rangle}_\pazocal{X} - {\left\langle {{\phi _ - },D} \right\rangle _\pazocal{X}} + {\left\langle {\mu ,\phi  - \nabla  \cdot \left( {{\rho _s} \cdot F_{\mathbf{u}}} \right)} \right\rangle _\pazocal{X}} \\
  - {\left\langle {{\lambda_0}, - {\rho _s}} \right\rangle _\pazocal{X}}-{\left\langle {{\lambda_1}, g\circ\mathbf{u}} \right\rangle _\pazocal{X}}+{\left\langle \rho_s-\rho^{\max}, \sigma\mathds{1}_{\pazocal{X}_d} \right\rangle _\pazocal{X}},
\end{array}
\end{equation}
where $\circ$ denotes function composition and $\sigma:\pazocal{X}\to\mathbb{R}_+$ is the Lagrange multiplier associated with the safety constraint. By Theorem \ref{thm:opt_con_duality}, when fixing the dual variable $\sigma$, the primal problem becomes
\begin{equation}\label{eq:constrained_primal_opt}
  \begin{aligned}
{J_p } = &{\left\langle {{\phi _ + },{V^{{{\mathbf{u}}^ \star }}}} \right\rangle _\pazocal{X}}\\
\rm{s.t.}\;&{\mathbf{u} }(x) = \mathop {\arg \min }\limits_{u \in \pazocal{U}} \;\nabla V \cdot F(x,u)+C(x,u)\\
&C_{\mathbf{u}} + \sigma\mathds{1}_{\pazocal{X}_d}+\nabla V \cdot F_{\mathbf{u}} = 0\\
&V{\mid _{{\pazocal{X}_g}}} = D.
\end{aligned}
\end{equation}
The only difference comparing the unconstrained case is the perturbation term $\sigma$ on the running cost within $\pazocal{X}_d$. Therefore, a primal-dual algorithm for the constrained optimal control can be formulated by alternating between solving the primal problem and updating the dual variable.
\begin{algorithm}[H]
    \caption{Primal-dual algorithm for optimal control with safety constraint}
    \label{alg:primal_dual_opt_con}
    \begin{algorithmic}[1] 
            \State  $\sigma[0] \gets \mathbf{0}$, $j=0$
            \Do
                \State Solve \eqref{eq:constrained_primal_opt} with $\sigma[j]$, get $\mathbf{u}^\star$.
                \State Compute stationary density $\rho_s$ under $\mathbf{u}^\star$.
                \State ${\sigma [j+1]} \gets \max\left\{\mathbf{0},{\sigma[j]} + \alpha \left( {(\rho _s-\rho^{\max})   \mathds{1}_{\pazocal{X}_d}} \right)\right\}$.
                \State $j\gets j+1$
            \doWhile { $\left(\begin{array}{l}
                       \left\langle {\sigma[j],\max(\mathbf{0},\rho _s-\rho^{\max})\mathds{1}_{\pazocal{X}_d}} \right\rangle > \epsilon  \\
                       \textbf{or} \left\|\rho_s\mathds{1}_{\pazocal{X}_d}\right\|_\infty>\rho^{\max}
                     \end{array}\right)$ }
            \State \textbf{return} $\mathbf{u}^\star,\rho_s,V$
    \end{algorithmic}
\end{algorithm}
Algorithm \ref{alg:primal_dual_opt_con} shows the primal-dual algorithm for constrained optimal control, where $\alpha>0$ is the step size and $\epsilon>0$ is the tolerance on the complementary slackness condition. The algorithm iterates between the primal value function problem, which solves for the optimal value function and control strategy under the perturbation $\sigma$, and the dual problem, which computes the stationary density function and updates the dual variable $\sigma$. It terminates if a feasible solution that is close enough to the optimal solution (assessed by the KKT condition) is found.
\subsection{Extension to systems with disturbance}
The primal-dual algorithm can be extended to solve constrained optimal control for systems with disturbance.

Consider the following dynamic system:
\begin{equation*}\label{eq:dyn_dist}
  \dot{x}=F(x,u,d),x\in\pazocal{X}=\mathbb{R}^n,u\in\pazocal{U}\subseteq\mathbb{R}^m,d\in\pazocal{D}\subseteq\mathbb{R}^p,
\end{equation*}
where $x$, $u$, and $d$ are the state, control input, and disturbance input of the system. We would like to solve the optimal control problem with the same state constraint as in \eqref{eq:constrained_opt_con_dual} under all possible disturbances. First, consider the worst-case disturbance for a given controller $\mathbf{u}$.

\begin{prop}
  For a fixed state-feedback controller $\mathbf{u}:\pazocal{X}\to\pazocal{U}$, consider the following optimal control problem:
  \begin{equation}\label{eq:worst_d_opt}
  \begin{aligned}
\mathop {\max }\limits_{d(\cdot)} &\int_0^T  {{\mathds{1}_{{\pazocal{X}_d}}(x(t))}dt} \\
\mathrm{s.t.}&~~\dot x(t) = F\left( {x(t),\mathbf{u}(x(t)),d(t)} \right),
\end{aligned}
\end{equation}
where $T$ is the first time that $x$ reaches $\pazocal{X}_g$. Then the state will not enter the danger set $\pazocal{X}_d$ iff the optimal value of \eqref{eq:worst_d_opt} is zero, and the worst-case disturbance is a function of the state only, i.e., $d^\star(t)=\mathbf{d}^\star(x(t))$.
\end{prop}
\begin{proof}
  The optimization in \eqref{eq:worst_d_opt} is essentially an optimal control problem in the Bolza form. Therefore, by the principle of optimality, for a single state $x$, the optimal strategy for the disturbance should be deterministic, i.e., we can parameterize the optimal (worst-case) $d^\star$ as $d(t)=\mathbf{d}^\star(x(t))$. When the optimal value of \eqref{eq:worst_d_opt} is zero, by the definition of the cost function, the state never enters $\pazocal{X}_d$, otherwise the state enters $\pazocal{X}_d$ under the worst-case disturbance.
\end{proof}

To emphasize the dependence of $\mathbf{d}^\star$ on $\mathbf{u}$, we denote $\mathbf{d}^\star[\mathbf{u}]$ as the worst case disturbance under control strategy $\mathbf{u}$.
By the duality result in Theorem \ref{thm:opt_con_duality}, \eqref{eq:worst_d_opt} can be rewritten as a density optimization. Combining this with \eqref{eq:constrained_opt_con_dual}, the robust density function optimization takes the following form:

\begin{equation}\label{eq:density_safety_opt1}
\begin{aligned}
\mathop {\min }\limits_{\bf{u},{\rho _s}}& \left\langle {{C_\mathbf{u}},{\rho _s}} \right\rangle - \left\langle {D,\phi_-} \right\rangle  \\
\mathrm{s.t.} &\left\{\begin{aligned}\mathop {\max }\limits_{\bf{d}} &\left\langle {\mathds{1}_{\pazocal{X}_d},{\rho _s}} \right\rangle  \\
\mathrm{s.t.}& \nabla  \cdot \left( {{\rho _s} \cdot F_\mathbf{u,d}} \right) = \phi,\mathbf{d}(x)\in\pazocal{D}\end{aligned}\right\}\le0,\\
&\forall x \in {\pazocal{X}},{\bf{u}}(x) \in {\pazocal{U}},
\end{aligned}
\end{equation}
where $F_{\mathbf{u},\mathbf{d}}\doteq F(x,\mathbf{u}(x),\mathbf{d}(x))$. The optimization in \eqref{eq:density_safety_opt1} solves for a controller $\mathbf{u}$ that robustly satisfies the safety constraint under the worst-case disturbance and optimizes the cost under the worst-case disturbance. Of course, when the disturbance changes, the controller is not guaranteed to be optimal, but the satisfaction of the safety constraint is still guaranteed.

To solve \eqref{eq:density_safety_opt1}, we simply need to add an additional step to the primal-dual algorithm to solve for the worst-case disturbance under the control strategy $\mathbf{u}$, which is illustrated in \eqref{eq:worst_d_opt}. The primal optimal control problem is simply modified from \eqref{eq:constrained_primal_opt} by adding the disturbance strategy in:
\begin{equation}\label{eq:constrained_pert_primal_opt}
  \begin{aligned}
{J_p } = &{\left\langle {{\phi _ + },{V^{{{\mathbf{u}}^ \star }}}} \right\rangle _\pazocal{X}}\\
\rm{s.t.}\;&{\mathbf{u} }(x) = \mathop {\arg \min }\limits_{u \in \pazocal{U}} \;\nabla V \cdot F(x,u,\mathbf{d}(x))+C(x,u)\\
&C_{\mathbf{u}} + \sigma\mathds{1}_{\pazocal{X}_d}+\nabla V \cdot F_{\mathbf{u},\mathbf{d}} = 0\\
&V{\mid _{{\pazocal{X}_g}}} = D.
\end{aligned}
\end{equation}
The primal-dual algorithm for systems with disturbance is shown as Algorithm \ref{alg:primal_dual_opt_con_robust}. For details, see \cite{chen2019optimal}.
\begin{algorithm}
    \caption{Primal-dual algorithm for robust safe control synthesis}
    \label{alg:primal_dual_opt_con_robust}
    \begin{algorithmic}[1] 
            \State  $\sigma[0] \gets \mathbf{0}$, $j=0$
            \Do
                \State Solve \eqref{eq:constrained_pert_primal_opt} with $\sigma[j]$, get $\mathbf{u}^\star$.
                \State Solve \eqref{eq:worst_d_opt} with $\mathbf{u}^\star$ to get the worst case $\mathbf{d}^\star[\mathbf{u}^\star]$
                \State Compute stationary density $\rho_s$ under $\mathbf{u}^\star$ and $\mathbf{d}^\star[\mathbf{u}^\star]$.
                \State ${\sigma [j+1]} \gets \max\left\{\mathbf{0},{\sigma[j]} + \alpha \left( {\rho _s  \mathds{1}_{\pazocal{X}_d}} \right)\right\}$.
                \State $j\gets j+1$
            \doWhile { $\left(\begin{array}{l}
                       \left\langle {\sigma[j],\max(\mathbf{0},\rho _s-\rho^{\max})\mathds{1}_{\pazocal{X}_d}} \right\rangle > \epsilon  \\
                       \textbf{or} \left\|\rho_s\mathds{1}_{\pazocal{X}_d}\right\|_\infty>\rho^{\max}
                     \end{array}\right)$ }
            \State \textbf{return} $\mathbf{u}^\star,\mathbf{d}^\star[\mathbf{u}^\star],\rho,V$

    \end{algorithmic}
\end{algorithm}
Note that the worst-case disturbance $\mathbf{d}^\star[\mathbf{u}^\star]$ is computed after the controller $\mathbf{u}^\star$, this ensures that the solution is robustly safe.

\subsection{Computation method}
The implementation of the primal-dual algorithm requires solving the HJB PDE and evaluate the density function multiple times. Since the state space and input space are continuous, the computation of the value function optimization and density function are done by Finite Element Method (FEM), i.e., gridding the state space.

\newsec{Implicit method for solving HJB PDE.}
 The standard way to solve the HJB PDE is by numerically integrate the PDE in FEM fashion, an alternative is to directly solve for the value function in an implicit fashion by solving linear equations, which leads to significant speed-up in some cases.

To be specific, let $x_g$ denote the array of grid points and $V_g$ denote the array of $V$ values on $x_g$. The partial differentials become differences between neighbors in $V_g$. We follow the upwind scheme to improve the stability of the FDM computation\cite{sun2015convergence}. Let $I=(i_1,...,i_n)$ be the index of a grid point, and $I_{-,k}=(i_1,...,i_k-1,...,i_n)$, $I_{+,k}=(i_1,...,i_k+1,...,i_n)$ are the indices of the two neighbors on the $k$-th dimension. Then define the backward and forward differences as
\begin{equation}\label{eq:upwind_deriv}
  \nabla^+ V(x_g^I)=\begin{bmatrix}
                      \frac{V_g^{I}-V_g^{I_{-,1}}}{\Delta x_1} \\
                      ... \\
                      \frac{V_g^{I}-V_g^{I_{-,n}}}{\Delta x_n}
                    \end{bmatrix}, \nabla^- V(x_g^I)=\begin{bmatrix}
                      \frac{V_g^{I_{+,1}}-V_g^{I}}{\Delta x_1} \\
                      ... \\
                      \frac{V_g^{I_{+,n}}-V_g^{I}}{\Delta x_n}
                    \end{bmatrix}
\end{equation}
When $I_{+,k}$ or $I_{-,k}$ are out of bounds of the grid on the boundary, the corresponding entries of the backward and forward differences are simply zero. Then under the upwind scheme, $\nabla V\cdot f(x_g^I)$ is computed as
\begin{equation}\label{eq:upwind}
\resizebox{1\hsize}{!}{$
  \nabla V\cdot f(x_g^I) \approx \max\{\mathbf{0},f(x_g^I)\}\cdot\nabla^+ V(x_g^I) + \min\{\mathbf{0},f(x_g^I)\}\cdot\nabla^- V(x_g^I),
  $}
\end{equation}
where the $\max$ and $\min$ are taken entry-wise.

Take the HJB PDE in \eqref{eq:constrained_primal_opt} as an example. When fixing the control strategy $\mathbf{u}$, the HJB PDE is turned into a difference equation of $V_g$:
 \begin{equation}\label{eq:implciit_HJB}
 \begin{aligned}
   &C_{\mathbf{u}}(x_g)+\sigma\mathds{1}_{\pazocal{X}_d}(x_g)+\nabla V_g F_{\mathbf{u}}(x_g)=0,&x&\notin\pazocal{X}_g\\
   &V_g(x_g) = D(x_g),&x&\in\pazocal{X}_g,
   \end{aligned}
 \end{equation}
 which is merely a linear algebraic equation of $V_g$, and is solved with indirect method (direct method is too slow and requires too much memory) such as first order gradient descent.
 Once $V_g$ is solved, the control input at each grid point $x_g^I$ can be evaluated as
 \begin{equation}\label{eq:xg_input}
   \mathbf{u}(x_g^I)  = \mathop {\arg \min }\limits_{u \in \pazocal{U}} \;\nabla V(x_g) \cdot F(x_g^I,u)+C(x_g^I,u),
 \end{equation}
 and we use linear interpolation between grid points to get $\mathbf{u}(x)$.
 The algorithm alternates between solving \eqref{eq:implciit_HJB} for $V_g$  and updating the controller $\mathbf{u}$ based on \eqref{eq:xg_input} until convergence.

\newsec{Kernel density estimation for density function estimation.}
The standard method for computing the density function is the two-step ODE procedure, presented in Section \ref{sec:density_review}, which gives accurate evaluation of the density function independent of the grid size. However, as the dimension of the state space gets high, the number of ODEs to be solved grows exponentially with the state dimension. An alternative is to approximate the density function using the kernel density estimation technique, which simulates the system for multiple trials and evaluate the density function with the samples collected via a kernel. The complexity is proportional to the number of trials instead of exponential w.r.t. the state dimension. An additional advantage is that it only needs a simulator instead of requiring an explicit model of the dynamic system.

Take the finite horizon case disucssed in Section \ref{sec:opt_control_destination} as an example. Given $\phi_+$, we sample $N$ initial conditions $x_0$ according to $\phi_+$ and for each $x_0$, simulate the system under $\mathbf{u}$, resulting in state sequence $\left\{x(t)\right\},t=0,t_s,2t_s,..., M t_s$, where $M$ should be large enough so that the state reaches $\pazocal{X}_g$ before $M t_s$; Each sample outside $\pazocal{X}_g$ then bears mass
\begin{equation}\label{eq:mass}
  m(x(t))= \frac{t_s \Phi_+}{N},
\end{equation}
where $\Phi_+=\int_{\pazocal{X}}{\phi_+(x) dx}$ denotes the total supply rate. To get the density at $x$, one simply compute
\begin{equation}\label{eq:kernel_density}
  \hat{\rho}_s(x)=\sum\limits_{i=1}^{N}\sum\limits_{j=1}^{M} K_h(x-x^i(j t_s)) m(x^i(j t_s)),
\end{equation}
where $K_h$ is the kernel with bandwidth $h$ and $x^i(j t_s)$ denotes the $j$-th state sample in the $i$-th trial. Here we choose the Epanechnikov kernel, which is defined as
\begin{equation}\label{eq:Epanechnikov}
  K_h(s) = \left\{ {\begin{aligned}
&\frac{3}{{4h}}\left( {1 - \frac{{{s^2}}}{{{h^2}}}} \right),&\left|s\right|\le h\\
&0&\left|s\right|> h
\end{aligned}} \right..
\end{equation}
It satisfies the requirement for a kernel, i.e., $\int_{-\infty}^\infty K_h(s)ds=1$ for all $h>0$, as shown in Fig. \ref{fig:Epan}.
\begin{figure}
  \centering
  \includegraphics[width=0.7\columnwidth]{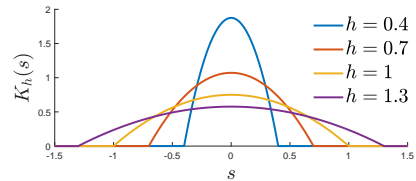}
  \caption{Epanechnikov kernel}\label{fig:Epan}
  \vspace{-0.8cm}
\end{figure}
 For $s\in\mathbb{R}^n$, one can simply take the product of $n$ Epanechnikov kernels to get an $n$-dim kernel:
\begin{equation*}
  K_\mathbf{h}(s)=\prod\nolimits_{i=1}^n K_{h_i}(s_i),
\end{equation*}
where $\mathbf{h}\in\mathbb{R}_{>0}^n$ is the bandwidth vector.
 Moreover, the Epanechnikov kernel is a kernel with finite support, which implies
 \begin{equation*}
   \forall s\notin \left\{s\mid \left|s_i\right|<h_i\right\},K_\mathbf{h}(s)=0.
 \end{equation*}
 This is particularly useful for density function estimation for safe control synthesis, since it gives a compact neighborhood of the query point $x$ such that any sample outside the neighborhood will not affect the density estimation of $x$. Suppose we use a kernel with infinite support, such as the Gaussian kernel, the density inside the danger set would never be zero and the primal-dual algorithm would not be implementable.

See \cite{densityexper2019} for more detailed analysis of the kernel density estimation.

\section{Optimal policy for MDPs with state density constraint}\label{sec:MDP_constrained}
Similar to the constrained optimal control for dynamic systems, the optimal policy for MDPs with state density constraint can be solved with the duality relationship as well. We propose a primal-dual algorithm that can not only solve the constrained optimal strategy for an MDP with density constraint, but also extends to multiple MDPs with aggregate density constraint.
\subsection{Constrained optimization with a single MDP}

The following constrained MDP optimization in the dual form is considered:
\begin{equation}\label{eq:dual_constrained_MDP}
  \begin{aligned}
J_d^\star=\max\limits_{\rho_s,\pi} &\sum\limits_s {{\rho_s(s)}\sum\limits_{a}\pi(a|s)\sum\limits_{s'}P_a(s,s') R_a(s,s')} \\
\rm{s.t.}\;&\overline{\rho}_s=\gamma (P^\pi)^\intercal \overline{\rho}_s+\overline{\phi}_+,\forall s,\pi(s)\in\Delta_s,\overline{\rho}_s\le\rho^{\max},
\end{aligned}
\end{equation}
where $\rho^{\max}\in\mathbb{R}^N$ is the upper bound of the density. Here we only consider an upper bound, extending this formulation to include a lower bound on the density is straightforward.
\begin{rem}
  We present the case without a sink, since the proposed method can handle the case with a sink by simply modifying the transition probability matrix $P^\pi$ to $\tilde{P}^\pi$, as discussed in Section \ref{sec:duality_MDP}.
\end{rem}
Similar to the optimal control case, we start from the Lagrangian. With the density constraint, the Lagrangian becomes
\begin{equation}\label{eq:MDP_lagrangian}
\begin{aligned}
  \pazocal{L}&=\sum\limits_s\rho_s(s)\left(
    \begin{array}{c}
    \sum\limits_a\pi(a|s)\sum\limits_{s'}P_a(s,s')R_a(s,s')-V(s) \\
    +\gamma\sum\limits_a\sum\limits_{s'}\pi(a|s)P_a(s,s')V(s')+\sigma(s)
  \end{array}
  \right)\\
  &+\sum\limits_s V(s) \phi_+(s),
\end{aligned}
\end{equation}
where $\sigma:S\to\mathbb{R}_{\ge0}$ is the lagrange multiplier corresponding to the constraint, and we let $\overline{\sigma}=[\sigma(s_1),...,\sigma(s_N)]^\intercal$ be its vector form.
When fixing $\sigma$, the primal optimization becomes
\begin{equation}\label{eq:primal_constrained_MDP}
\resizebox{1\hsize}{!}{$
\begin{aligned}
  J_p^\star &= {\left\langle {{\phi_+ },V} \right\rangle _S},\\
  \mathrm{s.t.}\;V\left( s \right) &= \sum\limits_{a\in A}\sum\limits_{s'\in S} {\pi(a|s) P_a\left( {s,s'} \right)\left( {{R_a}\left( {s,s'} \right) + \gamma V\left( {s'} \right)+\sigma(s)} \right)},\\
   \pi(s)& = \mathop {\arg \max }\limits_{\pi(s)\in\Delta_s} \left\{ {\sum\limits_{s'\in S} {{P_a}\left( s,s' \right)\left( {{R_a}\left( s,s' \right) + \gamma V\left( {s'} \right)} \right)} } \right\}.
  \end{aligned}
  $}
\end{equation}

Similar to the dynamic system case, a primal-dual algorithm solves the constrained MDP problem. The algorithm follows standard primal-dual update procedure based on the Lagrangian in \eqref{eq:MDP_lagrangian}. For simplicity, define the advantage function as
\begin{equation*}
\begin{aligned}
  \pazocal{A}(a|s)&=\sum\limits_{s'}P_a(s,s')(R_a(s,s')+\gamma V(s'))-V(s)\\
  \pazocal{A}(s)&=[\pazocal{A}(a_1|s),\pazocal{A}(a_2|s),...,\pazocal{A}(a_M|s)],
  \end{aligned}
\end{equation*}
and let $R^\pi$ be the reward vector under policy $\pi$, where
\begin{equation*}
  R^\pi(s)=\sum\limits_a \pi(a|s)\sum\limits_{s'}P_a(s,s')R_a(s,s').
\end{equation*}
The advantage function describes the expected benefit (comparing to the current value function) of taking action $a$ at state $s$, and $R^\pi$ describes the expected immediate reward following policy $\pi$.
The primal-dual algorithm is shown in Algorithm \ref{alg:primal_dual_MDP}, where $\alpha>0$ and $\beta>0$ are step sizes for the policy update and $\sigma$ update.  $\proj$ is the projection operator, which is implemented with quadratic programming:
\begin{equation*}\label{eq:proj_MDP}
  \proj_{\Delta_s}\pi(s)=\mathop {\arg \min }\limits_{c\in {\Delta _s}} \left\| {c - \pi(s)} \right\|_2^2.
\end{equation*}
\begin{algorithm}[H]
    \caption{Primal-dual algorithm for constrained MDP}
    \label{alg:primal_dual_MDP}
    \begin{algorithmic}[1] 
            \State  $\overline{\sigma}[0] \gets \mathbf{0}$, $j\gets 0$, $\forall s\in S,\pi(s)[0]\gets \pi_0(s)$
            \State $\overline{\rho}_s=(I-(\gamma (P^{\pi[j]})^\intercal))^{-1}{\overline{\phi}_+}$
            \Do
                \For{$s\in S$}
                    \State \resizebox{.75\hsize}{!}{$\pi(s)[j + 1] \gets \proj_{\Delta _s} \left( \pi(s)[j] + \alpha\rho_s(s) \pazocal{A}(s) \right)$}
                \EndFor
                \State $\overline{\rho}_s=(I-(\gamma (P^{\pi[j+1]})^\intercal))^{-1}{\overline{\phi}_+}$
                \State $\overline{V}=(I-\gamma P^{\pi[j+1]})^{-1} (R^\pi+\overline{\sigma})$
                \State ${\overline{\sigma}[j+1]} \gets \left\{ {\mathbf{0},{\overline{\sigma} [j]} + \beta \left( {\overline{\rho}_s - {\rho ^{\max }}} \right)} \right\}$
                \State $j \gets j+1$.
            \doWhile { \resizebox{.85\hsize}{!}{$\neg\left\{({\overline{\rho}_s } \le {\rho ^{\max }})\; \textbf{and}\;  \forall s\in S,\left\| \pi(s)[j+1]-\pi(s)[j] \right\|_\infty \le \epsilon\right\} $ }}
            \State \textbf{return} $\pi[j]$, $\rho_s$, $V$
    \end{algorithmic}
\end{algorithm}
$\pi_0$ is the initial policy. The algorithm alternates between primal update and dual update. In the primal update, $\pi(s)$ is updated by taking gradient ascent on the Lagrangian. Then the density function and value function under the new policy is calculated with simple algebraic equations from \eqref{eq:primal_constrained_MDP} and \eqref{eq:dual_constrained_MDP}. Then $\sigma$ is updated with the gradient w.r.t. the Lagrangian and the algorithm iterates between these steps until an optimal and feasible (safe) solution is found, which is assessed by the KKT condition.
\subsection{Constrained optimization with multiple MDPs}
The same setup can be extended to the case with multiple MDPs. Consider the case where there are $K$ MDPs: $(S^k,A^k,P_a^k,R_a^k),k=1,...,K$, each with its own supply function $\phi_+^k$ and policy $\pi^k$. Let $\rho_s^k$ be the stationary density function for the $k$-th MDP and define the cumulative density
\begin{equation*}
  \rho_c(s)\doteq\sum\limits_{k=1}^K\rho_s^k(s),
\end{equation*}
with $\overline{\rho}_c$ being its vector form. The constraint is on the cumulative density:
\begin{equation*}
  \overline{\rho}_c\le\rho^{\max}.
\end{equation*}
In this case, the constrained optimization is written as
\begin{equation}\label{eq:dual_multi_constrained_MDP}
\resizebox{.95\hsize}{!}{$
  \begin{aligned}
J_d^\star=\max\limits_{\rho_s^{1:K},\pi^{1:K}} &\sum\limits_{k=1}^K\sum\limits_{s\in S^k} {{\rho_s^k(s)}\sum\limits_{a\in A^k}\pi^k(a|s)\sum\limits_{s'\in S^k}P^k_a(s,s') R^k_a(s,s')} \\
\rm{s.t.}\;&\forall k=1,...K,\;\overline{\rho}_s^k=\gamma (P^{\pi^k})^\intercal \overline{\rho}_s^k+\overline{\phi}^k_+,\\
&\forall s\in S^k,\pi^k(s)\in\Delta^k_s,\sum\limits_{k=1}^K\overline{\rho}_s^k\le\rho^{\max},
\end{aligned}
$}
\end{equation}
where $\Delta_s^k$ is the set of all stochastic policies at state $s$ for the $k$-th MDP. Since the constraint is on the cumulative density, there is only one dual variable $\sigma$ associated with the constraint. The primal-dual algorithm for optimizing multiple MDPs with density constraint is shown as Algorithm \ref{alg:primal_dual_multiple_MDP}.
\begin{algorithm}
    \caption{Primal-dual algorithm for multiple constrained MDPs}
    \label{alg:primal_dual_multiple_MDP}
    \begin{algorithmic}[1] 
            \State  $\overline{\sigma}[0] \gets \mathbf{0}$, $j\gets 0$, $\forall s\in S,\pi(s)[0]\gets \pi_0(s)$
            \For{$k=1,...,K$}
                \State $\overline{\rho}^k_s=(I-(\gamma (P^{\pi^k[j]})^\intercal))^{-1}{\overline{\phi}^k_+}$
            \EndFor
            \Do
                \For{$k=1,...,K$}
                    \For{$s\in S$}
                        \State \resizebox{.75\hsize}{!}{$\pi^k(s)[j + 1] \gets \proj_{\Delta^k _s} \left( \pi^k(s)[j] + \alpha\rho_s^k(s) \pazocal{A}^k(s) \right)$}
                    \EndFor
                    \State $\overline{\rho}_s^k=(I-(\gamma (P^{\pi^k[j+1]})^\intercal))^{-1}{\overline{\phi}^k_+}$
                    \State $V^k=(I-\gamma P^{\pi^k[j+1]})^{-1} (R^{\pi^k}+\overline{\sigma})$
                \EndFor
                \State ${\overline{\rho} _c} \gets \sum\nolimits_{k = 1}^K {{\overline{\rho}_s^k}} $
                \State ${\overline{\sigma}[j+1]} \gets \left\{ {\mathbf{0},{\overline{\sigma} [j]} + \beta \left( {{\overline{\rho}_c} - {\rho ^{\max }}} \right)} \right\}$
                \State $j \gets j+1$.
            \doWhile { \resizebox{.75\hsize}{!}{$\neg\left\{\begin{aligned}&({\overline{\rho}_c } \le {\rho ^{\max }})\;\textbf{and}\;
              \forall k=1,...,K,\forall s\in S^k,\\&\left\| \pi^k(s)[j+1]-\pi^k(s)[j] \right\|_\infty \le \epsilon \end{aligned}\right\}$ }}
            \State \textbf{return} $\pi^{1:K}$, $\rho_s^{1:K}$, $V^{1:K}$
    \end{algorithmic}
\end{algorithm}

\section{Results and examples}\label{sec:result}
In this section we present three examples. The first example is a simple robot navigation problem with a danger area and exogenous disturbance. The purpose of this example is to demonstrate and visualize the proposed method. The second example is a practical segway obstacle avoidance problem with experiment implementation. The third example is a macroscopic traffic control problem modelled as multiple MDPs, and we show the optimal solution to the problem with density constraint.
\subsection{Robot navigation with obstacle}
We use a robot navigation problem as example, where the robot follows a simple 2D kinetic model:
\begin{equation*}\label{eq:robot_dyn}
  \begin{aligned}
  \dot{x}&=u_1\\
  \dot{y}&=u_2.
  \end{aligned}
\end{equation*}
The destination $\pazocal{X}_g=\left\{x|\left\|x\right\|\le0.1\right\}$ is a small ball around the origin, and the running cost is only a function of the state: $C(x)=\mathbf{1}-\mathds{1}_{\pazocal{X}_g}$, i.e., $C(x)=1$ when $x$ is outside $\pazocal{X}_g$ and 0 otherwise. The input bound $\pazocal{U}=\left\{u\mid \left\|u\right\|\le0.5\right\}$, the positive supply $\phi_+$ is plotted in Fig. \ref{fig:phi_plus} and the red circled area is $\pazocal{X}_d$.
\begin{figure}
  \centering
  \includegraphics[width=0.7\columnwidth]{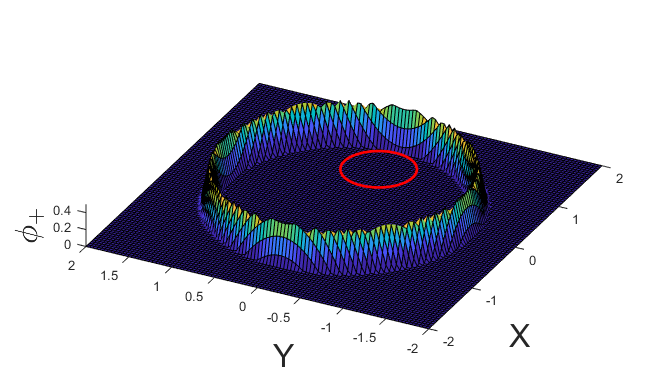}
  \caption{positive supply function $\phi_+$}\label{fig:phi_plus}
\end{figure}

The primal-dual algorithm terminates after 4 iterations, and a comparison of the result with and without the safety constraint is shown in Fig. \ref{fig:robot_res}.
\begin{figure}
  \centering
  \includegraphics[width=1.0\columnwidth]{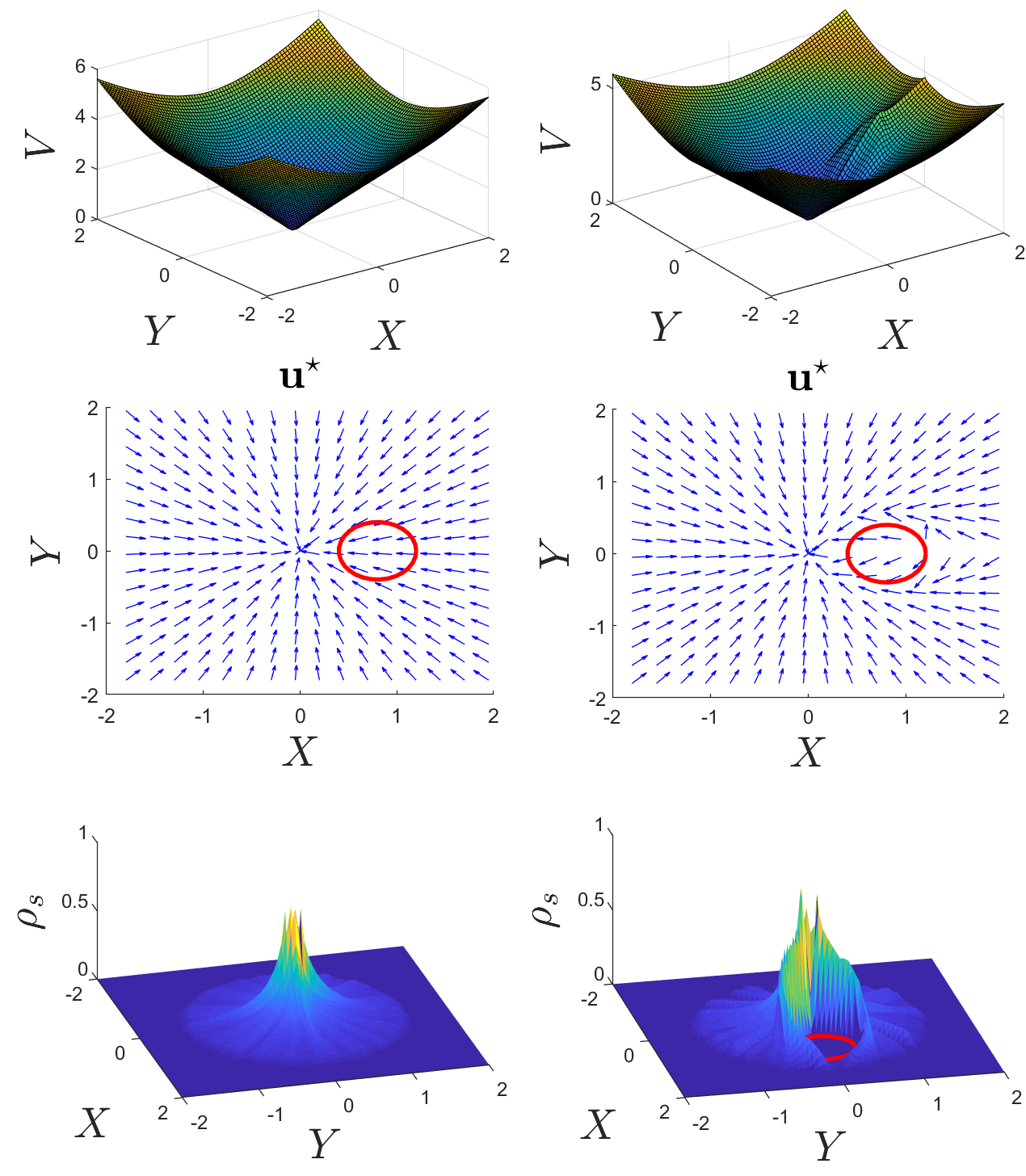}
  \caption{Constrained optimal control synthesis for 2D integrators with density function}\label{fig:robot_res}
\end{figure}

The left side shows the value function, the optimal control strategy, and the stationary density $\rho_s$ for the unconstrained case, and the right side shows the constrained case. From the density plot, we see zero density within the danger area $\pazocal{X}_d$. The value function of the constrained case has a small ``bump'' around $\pazocal{X}_d$, which is sufficient to steer all state around $\pazocal{X}_d$ and satisfy the safety constraint.

If we add disturbance to the dynamics as
\begin{equation*}
    \begin{aligned}
  \dot{x}&=u_1+d_1\\
  \dot{y}&=u_2+d_2,
  \end{aligned}
\end{equation*}
with $d\in\pazocal{D}=\left\{d\mid \left\|d\right\|\le0.25\right\}$, where $\pazocal{D}$ is half the size of $\pazocal{U}$, the robust primal-dual algorithm shown as Algorithm \ref{alg:primal_dual_opt_con_robust} is used to solve the constrained optimal control for this dynamic system with disturbance.
\begin{figure}
  \centering
  \includegraphics[width=1.0\columnwidth]{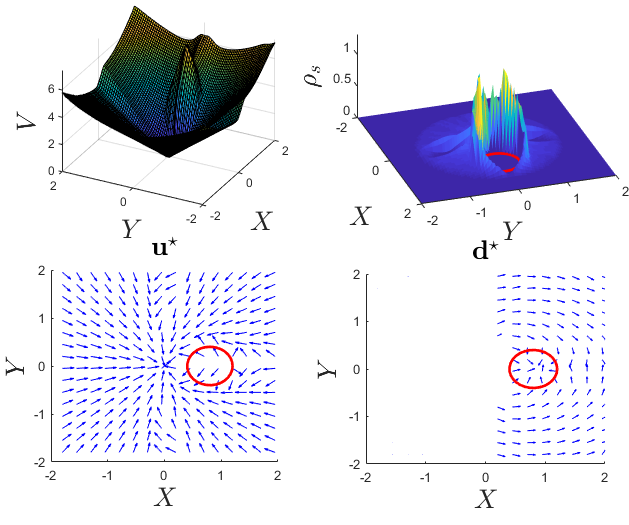}
  \caption{Synthesis result for 2D integrators with disturbance}\label{fig:robust_robot_res}
\end{figure}
Fig. \ref{fig:robust_robot_res}(a) shows the value function $V$, Fig. \ref{fig:robust_robot_res}(b) shows the density function $\rho_s$ under the worst-case disturbance $\mathbf{d}^\star$, where $\rho_s$ is zero everywhere inside the danger area $\pazocal{X}_d$. Fig. \ref{fig:robust_robot_res}(c) and Fig. \ref{fig:robust_robot_res}(d) shows the control strategy $\mathbf{u}^\star$ and the worst-case disturbance $\mathbf{d}^\star$. Comparing to the case without disturbance, the controller now is countering $\mathbf{d}^\star$ to make sure that no state enters the danger zone.

\subsection{Obstacle avoidance for a Segway}
In this section we present a practical example of segway obstacle avoidance. This example was presented in a conference paper \cite{densityexper2019}, and more detail can be found in the conference paper.
\begin{figure}[H]
\vspace{-0.2cm}
  \centering
  \includegraphics[width=0.85\columnwidth]{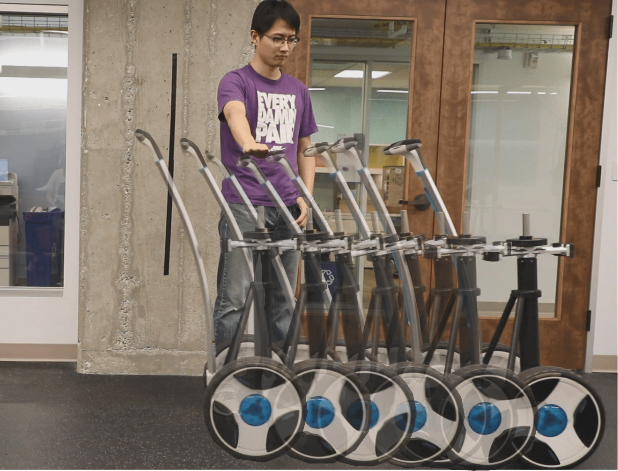}
  \caption{Segway experiment}\label{fig:seg_expr}
  \vspace{-0.5cm}
\end{figure}
The control objective is to move the segway from the initial position to the destination without hitting an obstacle, which has a round shape. The state of the segway is $x=[s,v,\theta,\dot{\theta}]^\intercal$, where $s$ and $v$ are the longitudinal position and velocity, $\theta$ and $\dot{\theta}$ are the pitch angle and angular velocity. The dynamics of the segway can be obtained via the Lagrangian method:
 \begin{equation}\label{eq:segway_dyn}
   \begin{bmatrix}
                   \dot{v} \\
                   \ddot{\theta}
                 \end{bmatrix}=\begin{bmatrix}
     M+m & ml\cos(\theta) \\
     ml\cos(\theta) & ml^2+J_e
   \end{bmatrix}^{-1}\begin{bmatrix}
                     u/R \\
                     mlg\theta -u
                   \end{bmatrix},
 \end{equation}
where $m$ is the cart mass, $l$ is the height of the CG of the cart, $R$ is the radius of the wheel, $g$ is the gravitational acceleration, $M$ is the equivalent mass of the wheels (accounting for the moment of inertia) and $J_e$ is the moment of inertia of the cart. $u$ is the control input, which is the torque applied on the wheel.

 Initially, the segway is equipped with a legacy controller $\mathbf{u}_0$ designed with Linear Quadratic Regulator (LQR) method:
 \begin{equation}\label{eq:LQR}
   \mathbf{u}_0(x) = K_1 \mathbf{Sat}_\eta(s-s_{des})+K_2v+K_3\theta +K_4\dot{\theta},
 \end{equation}
 where $K$ is obtained by solving the Riccati equation for the linearized dynamics, $s_{des}$ is the position of the destination, set to be $s_{des}=1.8m$, and $\mathbf{Sat}_\eta$ is the symmetric saturation function with range $[-\eta,\eta]$.

 \begin{figure}[H]
 \vspace{-0.3cm}
  \centering
  \includegraphics[width=0.8\columnwidth]{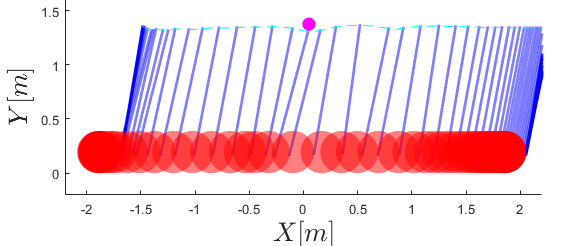}
  \caption{Experiment run of the segway}\label{fig:experiment_anim}
  \vspace{-0.5cm}
\end{figure}
 \begin{figure}[H]
  \centering
  \includegraphics[width=0.8\columnwidth]{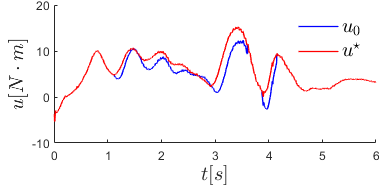}
  \caption{Input signals during the experiment}\label{fig:experiment_input}
  \vspace{-0.5cm}
\end{figure}
 The segway has a pole installed on one side, which would collide with the obstacle if following $\mathbf{u}_0$. We would like to find the optimal safe controller that is the closest to $\mathbf{u}_0$, which is achieved by posing the constrained optimal control problem with the following cost function:
 \begin{equation}\label{eq:cost}
 \resizebox{1\hsize}{!}{$
   V(x_0)=\int_0^\infty{e^{-\kappa t}\left\|u(t)-u_0(t)\right\|^2 dt},\mathrm{s.t.}x(0)=x_0,\dot{x}=F(x,u),
   $}
 \end{equation}
 where $\kappa$ is the discount factor and $F$ is the segway dynamics.

The primal-dual algorithm terminates after 223 iterations, outputting a safe controller. The safe controller is implemented on a customized segway platform in AMBER Lab at Caltech, as shown in Fig. \ref{fig:seg_expr}, and an experiment run is shown in Fig. \ref{fig:experiment_anim}. The segway accelerates before reaching the obstacle so that the tip got tilted down to avoid the obstacle. The dotted line shows the trajectory of the pole tip, which indeed avoids the obstacle. Fig. \ref{fig:experiment_input} shows the input signals during the experiment. The red plot denotes the control input of the safe controller $\mathbf{u}^\star$, and the blue plot denotes the original LQR controller $\mathbf{u}_0$. The video of the experiment can be found at \href{https://youtu.be/pdEtknFGu-A}{https://youtu.be/pdEtknFGu-A}.

\subsection{Traffic control with MDP}\label{sec:traffic_MDP}
In this section, we illustrate the constrained MDP optimization with a traffic control example from \cite{kouvelas2017enhancing}, where the task is to control the macroscopic traffic flow for the area shown in Fig. \ref{fig:map_traffic}. The area is divided into $N=7$ regions, and for each region, the traffic capacity is governed by the Macroscopic Fundamental Diagram (MFD) of traffic flow \cite{geroliminis2011properties}. The idea is that when the vehicle density is low, the traffic flow rate increases with vehicle density; when the vehicle density is larger than a threshold, congestion starts to form and the flow rate decreases with vehicle density. One example is shown in Fig. \ref{fig:MFD}, which uses the data from \cite{ji2018determining}.
\begin{figure}
  \centering
  \includegraphics[width=0.7\columnwidth]{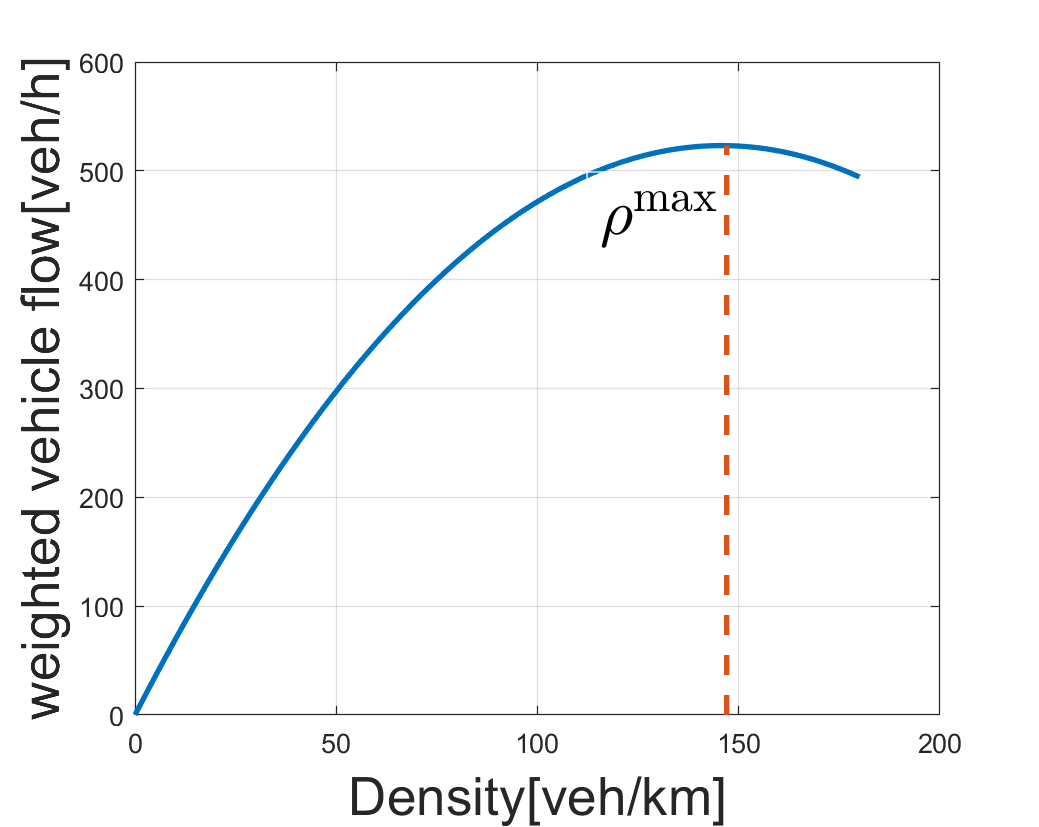}
  \caption{Macroscopic Fundamental Diagram}\label{fig:MFD}
\end{figure}
 We call the turning point the critical density, denoted as $\rho^{\max}$.
\begin{figure}[H]
  \centering
  \includegraphics[width=0.5\columnwidth]{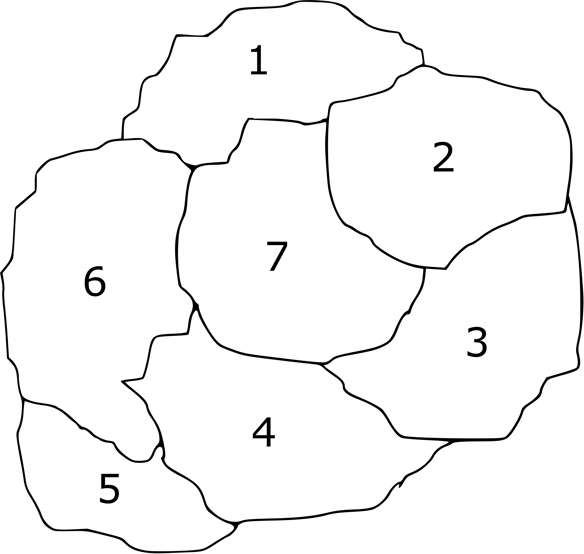}
  \caption{The map of the traffic control area}\label{fig:map_traffic}
\end{figure}
It is assumed that each region in Fig. \ref{fig:map_traffic} has a critical density $\rho _i^{\max }$, and the task of traffic control is to minimize the cost while keeping the density of each region below the critical density.

Since the task is to reach the destination with a path as short as possible, the goal is not to maximize a reward, but rather to minimize a cost. Since the cost can be understood as a negative reward, all the algorithms for MDP presented previously still apply, except for some changing of signs. We assume that for each vehicle, the transition cost only depends on the state. For example, a vehicle wants to get to region 1 from region 5, if it takes the path 5-4-7-1, then the total cost is $C_5+C_4+C_7+C_1$. The state cost is set as $C=[1.2,1.2,1.4,1.1,1,1.6,0.8]^T$.
$\gamma$ is chosen to be 1 since the vehicles won't exit the road until they reach their destinations, and the destinations are modeled as the sink set $S_-$.

 In fact, this problem cannot be solved as a single MDP since the vehicles have different destinations. It is solved as $K=7$ MDPs with 7 different destinations. Let $\phi_{+}^k(s_i)$ denote the traffic demand from $s_i$ to $s_k$, and we assume that all $\phi_{+}^k$ are given. The action is simply to choose which neighboring region to visit next. Therefore, the stochastic policy is simply in the form of transition probability matrices for the 7 MDPs. The set of stochastic policies at state $s_i$ is then simply
 \begin{equation}\label{eq:probability_simplex}
  {\Delta _{s_i}} = \left\{ {\alpha \in \mathbb{R}_{\ge0} ^{1 \times N}\mid \sum\nolimits_j {{\alpha_j}} =1,\left( {j \notin {\pazocal{N}_i}} \right) \to \left( {{\alpha_j} = 0} \right)\;} \right\},
\end{equation}
where $\pazocal{N}_i$ is the neighbor set of $s_i$.

 The density optimization problem is formulated as follows:
\begin{equation}\label{eq:traffic_MDP_density}
  \begin{aligned}
\mathop {\min }\limits_{\pi^{1:N} ,\rho_s^{1:N} } \sum\nolimits_{i = 1}^N &{{C_i}\sum\nolimits_{k = 1}^K {\rho_s(s_i)^k} }\\
\rm{s.t.}&\forall k\in\left\{1,...,K\right\},{\overline{\rho}_s ^k} = (\tilde{P}^{{\pi^k}})^\intercal{\overline{\rho}_s ^k} + {\overline{\phi}_+ ^k}\\
&\forall i\in\left\{1,...,N\right\},\sum\nolimits_{k = 1}^N {\rho_s ^k(s_i)}  \le \rho _i^{\max },
\end{aligned}
\end{equation}
where $\pi^k$ is the strategy for the traffic demand with destination $s_k$, which determines the transition probability matrix $P^{\pi^k}$. The cropped transition probability matrix $\tilde{P}^{\pi^k}$ is modified from $P^{\pi^k}$ by changing the $k$-th column to all zeros, as discussed in Section \ref{sec:duality_MDP}. $\overline{\rho}_s^k\in\mathbb{R}^N$ is the stationary traffic density vector with destination $s_k$ under $\pi^k$.

 As an example, we pose density constraint only on region 7, since it's in the center of the map and likely the most popular route to take. Without the constraint, the 7 MDPs are decoupled, and the optimal policy for each of them can be solved independently by value iteration or policy iteration. Moreover, according to \cite{puterman2014markov}, almost surely, the optimal policies would be deterministic. With the density constraint, the 7 MDPs are coupled, and are solved with Algorithm \ref{alg:primal_dual_multiple_MDP}. The comparison of the density distribution with and without the density constraint is shown in Fig. \ref{fig:MDP_plot} in color difference. The left plot is the cumulative density in the unconstrained case where a very high density appears in region 7; the right plot is the cumulative density in the constrained case where the density in region 7 is diverted to other regions. The overall cost for the unconstrained case is 71.05, and it increases to 79.21 for the constrained case.
\begin{figure}[H]
  \centering
  \includegraphics[width=0.9\columnwidth]{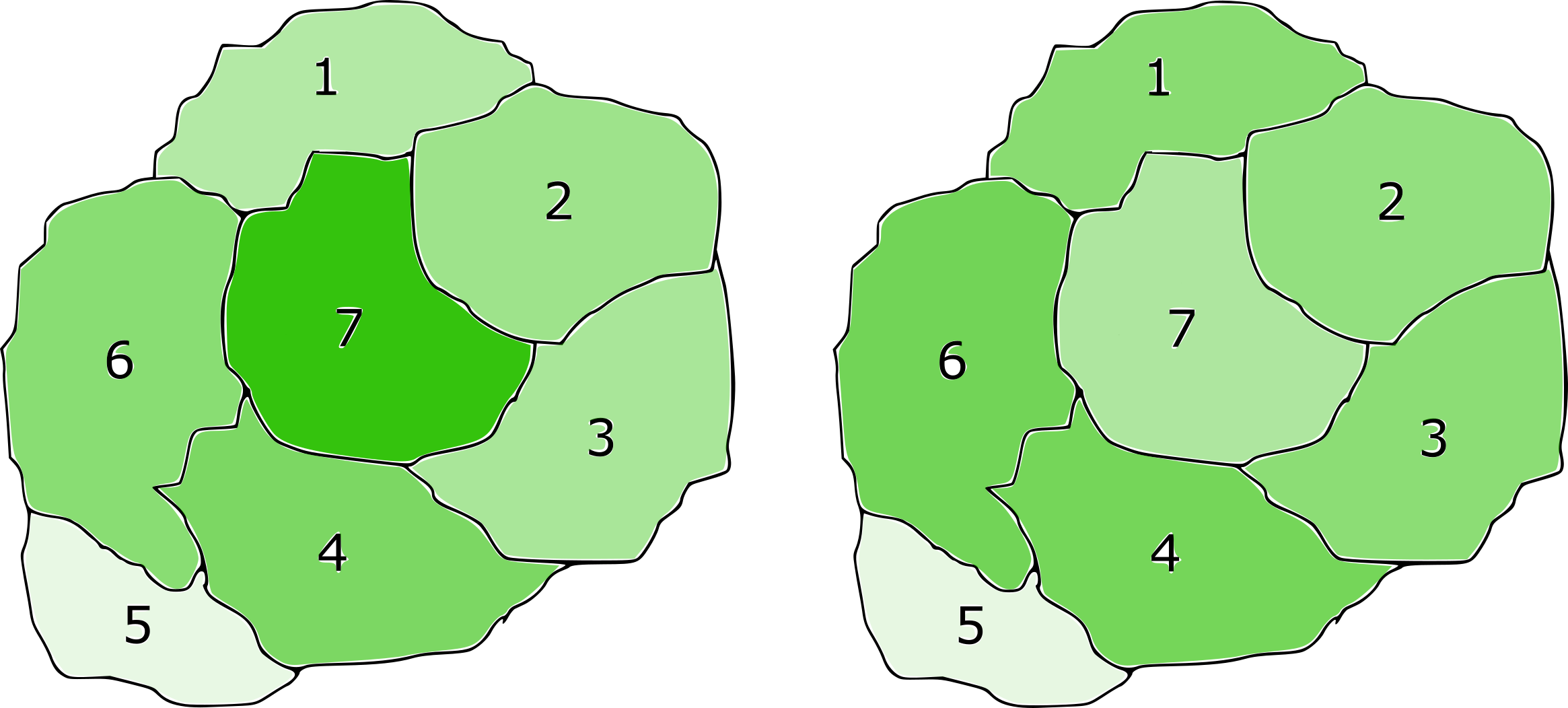}
  \caption{Comparison of the cumulative density in constrained and unconstrained MDP}\label{fig:MDP_plot}
\end{figure}

\begin{figure}[H]
  \centering
  \includegraphics[width=0.9\columnwidth]{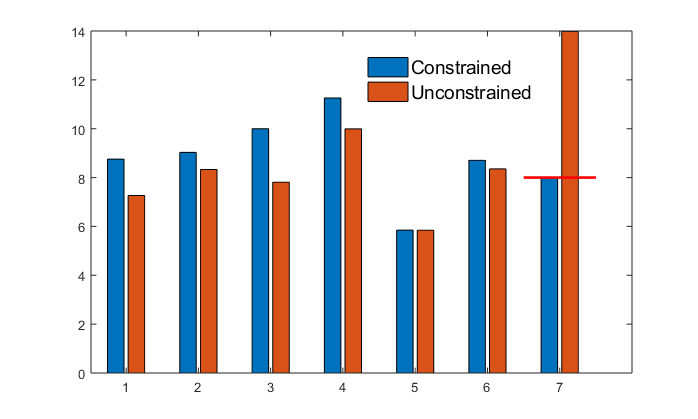}
  \caption{Bar plot of the density in constrained and unconstrained cases}\label{fig:traffic_histo}
\end{figure}
Fig. \ref{fig:traffic_histo} shows the bar plot in the two cases, where the red line on the 7-th column shows the upper bound of the cumulative density of region 7.

Different from the unconstrained case, there may not be a deterministic policy that is optimal. A simple example is an MDP with only two states, one has larger reward, but with a constraint on the density of it. Then the optimal policy is obviously a stochastic policy that barely satisfies the density constraint. Similarly, in the traffic example, the optimal policy for the constrained case turns out to be stochastic, rendering the density at $s_7$ exactly at $\rho_7^{\max}$.

\section{Conclusion and future works}\label{sec:conclusion}

In this paper, we present the density function as the dual of the value function in both optimal control and Markov decision process. Some constraint such as safety constraint and density constraint can then be formulated as an optimization on the density function. Then multiple primal-dual algorithms are proposed to solve the constrained optimal control problems with different setups. The density function approach is extended to the systems with disturbance, and to multiple MDPs coupled by density constraints. We demonstrate the capability of the formulation with two examples, one on robot navigation and one on macroscopic traffic control.

One interesting direction that we plan to try out is to combine the density function with reinforcement learning. There is a natural connection between reinforcement learning and optimal control and the popular methods such as Q learning \cite{watkins1992q} is essentially an approximation of the HJB PDE. However, the learning community typically assumes that a simulator is available rather than a explicit dynamic model, which is needed to compute the density function with the ODE approach. Fortunately, the kernel density estimation approach discussed in \cite{densityexper2019} is able to approximate the density function with just a simulator. We plan to combine reinforcement learning with density function approximation via kernel density estimation in the primal-dual algorithm proposed in this paper and enforce safety on learning. Other future works include parameterized optimal control to accelerate the computation, direct use of the density function in a finite horizon approximation of the optimal control solution (e.g. model predictive control), and analysis of the convergence of the primal-dual algorithm.

 \balance
\bibliographystyle{abbrv}
\bibliography{density_bib}

\begin{thebibliography}{10}

\bibitem{altman1999constrained}
E.~Altman.
\newblock {\em Constrained Markov decision processes}, volume~7.
\newblock CRC Press, 1999.

\bibitem{bellman2013dynamic}
R.~Bellman.
\newblock {\em Dynamic programming}.
\newblock Courier Corporation, 2013.

\bibitem{brockett2007optimal}
R.~W. Brockett.
\newblock Optimal control of the liouville equation.
\newblock {\em AMS IP Studies in Advanced Mathematics}, 39:23, 2007.

\bibitem{chen2019optimal}
Y.~Chen, M.~Ahmadi, and A.~D. Ames.
\newblock Optimal safe controller synthesis: A density function approach.
\newblock {\em arXiv preprint arXiv:1909.11798}, 2019.

\bibitem{densityexper2019}
Y.~Chen, A.~Singletary, and A.~D. Ames.
\newblock Density functions for guaranteed safety on robotic systems.

\bibitem{dai2017boosting}
B.~Dai, A.~Shaw, N.~He, L.~Li, and L.~Song.
\newblock Boosting the actor with dual critic.
\newblock {\em arXiv preprint arXiv:1712.10282}, 2017.

\bibitem{geroliminis2011properties}
N.~Geroliminis and J.~Sun.
\newblock Properties of a well-defined macroscopic fundamental diagram for
  urban traffic.
\newblock {\em Transportation Research Part B: Methodological}, 45(3):605--617,
  2011.

\bibitem{ji2018determining}
Y.~Ji, M.~Xu, J.~Li, and H.~J. van Zuylen.
\newblock Determining the macroscopic fundamental diagram from mixed and
  partial traffic data.
\newblock {\em Promet-Traffic\&Transportation}, 30(3):267--279, 2018.

\bibitem{kouvelas2017enhancing}
A.~Kouvelas, M.~Saeedmanesh, and N.~Geroliminis.
\newblock Enhancing model-based feedback perimeter control with data-driven
  online adaptive optimization.
\newblock {\em Transportation Research Part B: Methodological}, 96:26--45,
  2017.

\bibitem{lasserre2008nonlinear}
J.~B. Lasserre, D.~Henrion, C.~Prieur, and E.~Tr{\'e}lat.
\newblock Nonlinear optimal control via occupation measures and
  lmi-relaxations.
\newblock {\em SIAM journal on control and optimization}, 47(4):1643--1666,
  2008.

\bibitem{majumdar2014convex}
A.~Majumdar, R.~Vasudevan, M.~M. Tobenkin, and R.~Tedrake.
\newblock Convex optimization of nonlinear feedback controllers via occupation
  measures.
\newblock {\em The International Journal of Robotics Research},
  33(9):1209--1230, 2014.

\bibitem{pontryagin2018mathematical}
L.~S. Pontryagin.
\newblock {\em Mathematical theory of optimal processes}.
\newblock Routledge, 2018.

\bibitem{prajna2004nonlinear}
S.~Prajna, P.~A. Parrilo, and A.~Rantzer.
\newblock Nonlinear control synthesis by convex optimization.
\newblock {\em IEEE Transactions on Automatic Control}, 49(2):310--314, 2004.

\bibitem{puterman2014markov}
M.~L. Puterman.
\newblock {\em Markov decision processes: discrete stochastic dynamic
  programming}.
\newblock John Wiley \& Sons, 2014.

\bibitem{rantzer2001dual}
A.~Rantzer.
\newblock A dual to lyapunov's stability theorem.
\newblock {\em Systems \& Control Letters}, 42(3):161--168, 2001.

\bibitem{sun2015convergence}
B.~Sun and B.-Z. Guo.
\newblock Convergence of an upwind finite-difference scheme for
  hamilton--jacobi--bellman equation in optimal control.
\newblock {\em IEEE Transactions on Automatic Control}, 60(11):3012--3017,
  2015.

\bibitem{teschl2012ordinary}
G.~Teschl.
\newblock {\em Ordinary differential equations and dynamical systems}, volume
  140.
\newblock American Mathematical Soc., 2012.

\bibitem{watkins1992q}
C.~J. Watkins and P.~Dayan.
\newblock Q-learning.
\newblock {\em Machine learning}, 8(3-4):279--292, 1992.

\bibitem{zhao2017control}
P.~Zhao, S.~Mohan, and R.~Vasudevan.
\newblock Control synthesis for nonlinear optimal control via convex
  relaxations.
\newblock In {\em 2017 American Control Conference (ACC)}, pages 2654--2661.
  IEEE, 2017.

\bibitem{zhao2017optimal}
P.~Zhao, S.~Mohan, and R.~Vasudevan.
\newblock Optimal control for nonlinear hybrid systems via convex relaxations.
\newblock {\em arXiv preprint arXiv:1702.04310}, 2017.

\end{thebibliography}
\end{document}